\documentclass[twopage, 11pt]{article}

\usepackage{amsmath, amssymb, graphicx, amsthm}	
\usepackage{caption}
\usepackage{subcaption}
\RequirePackage[authoryear,square,comma]{natbib}
\newtheorem{thm}{Theorem}[section]

\newtheorem{lem}[thm]{Lemma}
\newtheorem{prop}[thm]{Proposition}

\newcommand{\ind}{{\bf 1}}
\newcommand{\PP}{{\mathbb P}}
\newcommand{\RR}{{\mathbb R}}
\newcommand{\NN}{{\mathbb N}}
\newcommand{\EE}{{\mathbb E}}
\newcommand{\E}{{\mathbb E}}

\newcommand{\nn}{\mbox{NN}}

\newcommand{\tr}{\mbox{tr}}

\newcommand{\inv}{^{-1}}

\def\calA{{\cal A}}

\def\calG{{\cal G}}

\newcommand{\gi}{{[g]}}

\newcommand{\argmin}{\operatornamewithlimits{argmin}}

\def\b#1{\left(#1\right)}
\def\bb#1{\left\{#1\right\}}
\def\bbb#1{\left[#1\right]}
\def\n#1{\left|#1\right|}
\def\nn#1{\left\|#1\right\|}

\begin{document}

\title{Network Granger Causality with Inherent Grouping Structure}
\date{}
\author{Sumanta Basu\\ Department of Statistics, University of Michigan
        \and Ali Shojaie\\ Department of Biostatistics, University of Washington
        \and George Michailidis\\ Department of Statistics, University of Michigan}


\maketitle

\begin{abstract}
The problem of estimating high-dimensional network models arises naturally in the analysis of many biological and socio-economic systems. 
In this work, we aim to learn a network structure from temporal panel data, employing the framework of Granger causal models under the assumptions of sparsity of its edges and inherent grouping structure among its nodes. To that end, we introduce a group lasso regression regularization
framework, and also examine a thresholded variant to address the issue of group misspecification. Further, the norm consistency
and variable selection consistency of the estimates are established, the latter under the novel concept of direction consistency. 
The performance of the proposed methodology is assessed through an extensive set of simulation studies and comparisons with existing techniques.
The study is illustrated on two motivating examples coming from functional genomics and financial econometrics.
\end{abstract}

\section{Introduction}

We consider the problem of learning a directed network of interactions among a number of entities from time course data. A natural framework to analyze this problem uses the notion of Granger causality \cite{granger69}. Originally proposed by C.W. Granger this notion provides a statistical framework for determining whether a time series $X$ is useful in forecasting another one $Y$, through a series of statistical tests. It has found wide applicability in economics, including testing relationships between money and income \citep{sims1972money}, government spending and taxes on economic output \citep{blanchard2002empirical}, stock price and volume \citep{hiemstra1994testing}, etc. More recently the Granger causal framework has found diverse applications in biological sciences including functional genomics, systems biology and neurosciences to understand the structure of gene regulation, protein-protein interactions and brain circuitry, respectively.  

It should be noted that the concept of Granger causality is based on associations between time series, and only under very stringent conditions, true causal relationships can be inferred \citep{pearl2000causality}. Nonetheless, this framework provides a powerful tool for understanding the interactions among random variables based on time course data.

Network Granger causality (NGC) extends the notion of Granger causality among two variables to a wider class of $p$ variables. Such extensions involving multiple time series are handled through the analysis of vector autoregressive processes (VAR) \citep{lutkepohl2005new}. Specifically, for $p$ stationary time series $X_1^t, \ldots, X_p^t$, with $\mathbf{X^t} = (X_1^t, \ldots, X_p^t)'$, one considers the class of models
\begin{equation}\label{NGCdefn}
	\mathbf{X^t} = A^1 \mathbf{X^{t-1}} + \ldots + A^d \mathbf{X^{t-d}} + \mathbf{\epsilon^t},
\end{equation}
with $d$ being the \textit{ unknown} order of the VAR model and the innovation process satisfying $\mathbf{\epsilon^T} \sim N(0, \sigma^2 I)$. We call $A^1, \ldots, A^d$ the adjacency matrices from lags $1, \ldots, d$.  In this model, the time series $\{X^t_j\}$ is said to be Granger causal for the time series $\{X^t_i\}$ if $A^h_{i,j} \neq 0$ for some $h = 1, \ldots, d$.  Equivalently we can say that there exists an edge $X^{t-h}_j \rightarrow X^t_i$ in the underlying network model comprising of $(d+1) \times p$ nodes (see Figure~\ref{GGCdemo}).
\begin{figure}[h]
\begin{center}
\includegraphics[scale = 0.5, clip = TRUE, trim = 1.3in 1.5in 2in 1.2in]{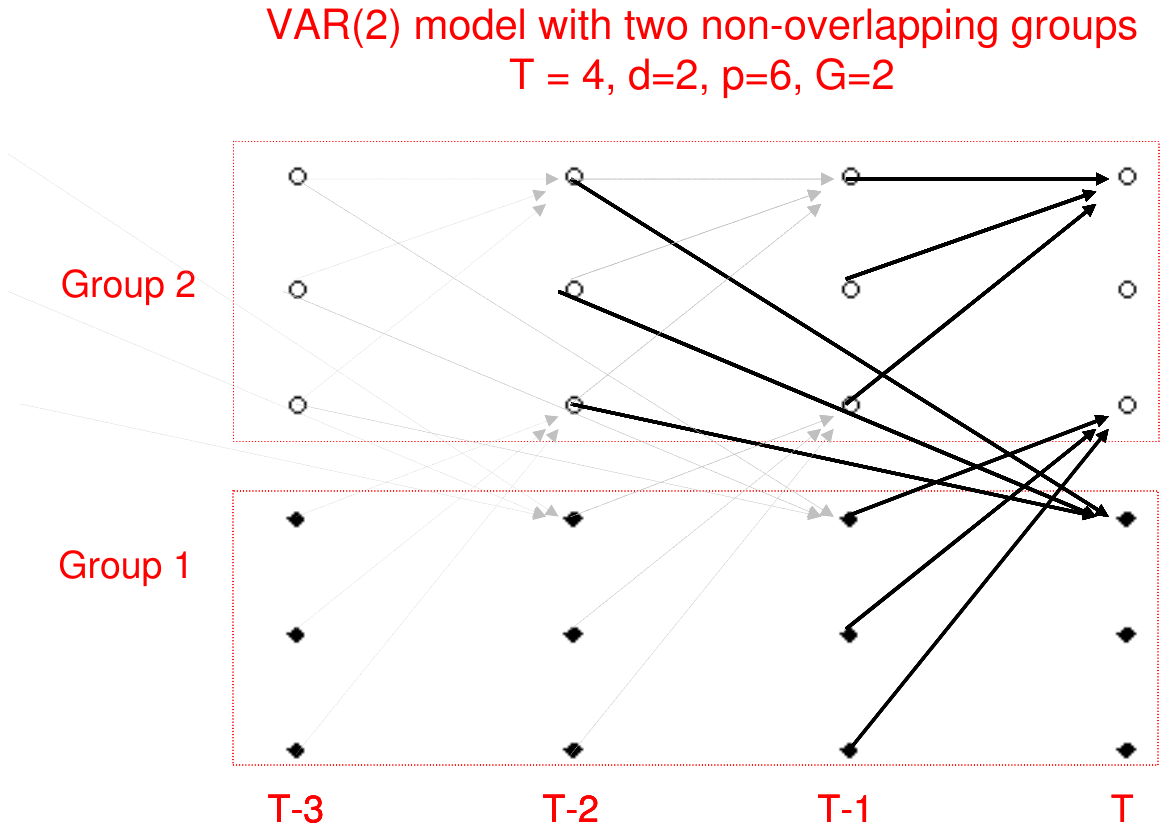}
\caption{An example of a Network Granger causal model with two non-overlapping groups observed over T = 4 time points}\label{GGCdemo}
\end{center}
\end{figure}
Note that the temporal structure induces a natural partial order among the nodes of this network, which in turn simplifies significantly the corresponding estimation problem \citep{alipendag10} of a directed acyclic graph. Nevertheless, one still has to deal with estimating a high-dimensional network (e.g. hundreds of genes) from a limited number of samples.

The traditional asymptotic framework of estimating VAR models requires observing a long, stationary realization $\{X^1, \ldots, X^T, \, T \rightarrow \infty \, p, \, d \mbox{ fixed}\}$ of the $p$-dimensional time series. This is not appropriate in many biological applications for the following reasons. First, long stationary time series are rarely observed in these contexts. Second, the number of time series ($p$) being large compared to $T$, the task of consistent order ($d$) selection using standard criteria (e.g., AIC or BIC) becomes challenging. 
Similar issues arise in many econometric applications where empirical evidence suggests lack of stationarity over a long time horizon, although the multivariate time series exhibits locally stable distributional properties. 

A more suitable framework comes from the study of panel data, where one observes several replicates of the time series, with possibly short $T$, across a panel of $n$ subjects. In biological applications replicates are obtained from test subjects. In the analysis of macroeconomic variables,  households or firms typically serve as replicates. After removing panel specific fixed effects one treats the replicates as independent samples, performs regression analysis under the assumption of common slope structure and studies the asymptotics under the regime $n \rightarrow \infty$.
Recent works of \citet{CaoSunJE2011} and \citet{binder2005} analyze theoretical properties of short panel VARs in the low-dimensional setting ($n \rightarrow \infty, \, T,\, p$ fixed).

The focus of this work is on estimating a {\textit{ high-dimensional NGC model}} in the panel data context ($p, \, n$ large, $T$ small to moderate). This work is motivated by two application domains, functional genomics and financial econometrics.
 In the first application (presented in Section ~\ref{secdata}) one is interested in reconstructing a gene regulatory network structure from time course data, a canonical problem in functional genomics \citep{michailidis2012}. The second motivating example examines the composition of balance sheets of the $n = 50$ largest US banks by size, over $T = 9$ quarterly periods, which provides insight into their risk profile. 

The nature of high-dimensionality in these two examples comes from both estimation of $p^2$ coefficients for each of the adjacency matrices $A^1, \ldots, A^d$, but also from the fact that the order of the time series $d$ is often unknown. Thus, in practice, one must either ``guess'' the order of the time series (often times, it is assumed that the data is generated from a VAR(1) model, which can result in significant loss of information), or include all of the past time points, resulting in significant increase in the number of variables in cases where $d \ll T$.
Thus, efficient estimation of the order of the time series becomes crucial.

Latent variable based dimension reduction techniques like principal component analysis or factor models are not very useful in this context since our goal is to reconstruct a network among the observed variables. To achieve dimension reduction we impose a group sparsity assumption on the structure of the adjacency matrices $A_1, \ldots, A_d$. In many applications, structural grouping information about the variables exists. For example, genes can be naturally grouped according to their function or chromosomal location, stocks according to their industry sectors, assets/liabilities according to their class, etc. This information can be incorporated to the Granger causality framework through a group lasso penalty. If the group specification is correct it enables estimation of denser networks with limited sample sizes  \citep{bach08, huangzhang10, lounici2011}. However, the group lasso penalty can achieve model selection consistency only at a group level. In other words, if the groups are misspecified, this procedure can not perform within group variable selection \citep{grpbridge09}, an important feature in many applications.

Over the past few years, several authors have adopted the framework of network Granger causality to analyze multivariate temporal data. For example, \citet{fujitaetal07} and \citet{lozanoetal09} employed NGC models coupled with penalized $\ell_1$ regression methods to learn gene regulatory mechanisms from time course microarray data.
Specifically, \citet{lozanoetal09} proposed to group all the past observations, using a variant of group lasso penalty, in order to construct a relatively simple Granger network model. This penalty takes into account the average effect of the covariates over different time lags and connects Granger causality to this average effect being significant. However, it suffers from significant loss of information and makes the consistent estimation of the signs of the edges difficult (due to averaging). \citet{alitrunc} proposed a truncating lasso approach by introducing a truncation factor in the penalty term, which strongly penalizes the edges from a particular time lag, if it corresponds to a highly sparse adjacency matrix.

Despite recent use of NGC in applications involving high dimensional data, theoretical properties of the resulting estimators have not been fully investigated. For example, \citet{lozanoetal09} and \citet{alitrunc} discuss asymptotic properties of the resulting estimators, but neither address in depth norm consistency properties, nor do they examine under what vector autoregressive structures the obtained results hold. 

In this paper, we develop a general framework that accommodates different variants of group lasso penalties for NGC models. It allows for the simultaneous estimation of the order of the times series and the Granger causal effects; further, it allows for variable selection even when the groups are misspecified. In summary, the key contributions of this work are:
(i) investigate in depth \textit{sufficient conditions} that explicitly take into consideration the structure of the VAR$(d)$ model to establish norm  consistency,
(ii) introduce the novel notion of \textit{direction consistency}, which generalizes the concept of sign consistency and provides insight into the properties of group lasso estimates within a group, and (iii) use the latter notion to introduce an easy to compute thresholded variant of group lasso, that performs within group variable selection in addition to group sparsity pattern selection {even when the group structure is misspecified}. 

All the obtained results are non-asymptotic in nature, which help provide insight into the properties of the estimates under different asymptotic regimes arising from varying growth rates of $T, p, n$, group sizes and the number of groups.

\section{Model and Framework}\label{secmodel}
\textbf{Notation. } Consider a VAR model
\begin{equation}\label{eqn1:NGCdefn}
\underbrace{\mathbf{X}^t}_{p \times 1} = \underbrace{A^1}_{p \times p} \mathbf{X}^{t-1} + \ldots + A^d \mathbf{X}^{t-d} + \mathbf{\epsilon}^t,~~\epsilon^t \sim N(\mathbf{0}, \sigma^2 \mathbf{I}_{p \times p})
\end{equation}
observed over $T$ time points $t = 1, \ldots, T$, across $n$ panels. 
The index set of the variables $\mathbb{N}_p = \{ 1, 2, \ldots, p\}$ can be partitioned into $G$ non-overlapping groups $\mathcal{G}_g$, i.e., $\mathbb{N}_p = \cup_{g=1}^G \mathcal{G}_g$ and $\mathcal{G}_g \cap \mathcal{G}_{g'} = \phi$ if $g \ne g'$. Also $k_g = |\mathcal{G}_g|$  denotes the size of the $g^{th}$ group with $k_{max} = \displaystyle \max_{1 \le g \le G} k_g$.

For any matrix $A$, we denote the $i^{th}$ row by $A_{i:}$, $j^{th}$ column by $A_{:j}$ and the collection of rows (columns)  corresponding to the $g^{th}$ group by $A_{[g]:}$ ($A_{:[g]}$). The transpose of a matrix $A$ is denoted by $A'$ and its Frobenius norm by $|| A ||_{F}$. The symbol $A^{1:T}$ is used to denote the concatenated matrix $\left[A^1: \cdots: A^T \right]$. For any matrix or vector $D$, $\| D \|_0$ denotes the number of non-zero coordinates in $D$.  For notational convenience, we reserve the symbol $\|. \|$ to denote the $\ell_2$ norm of a vector and/or the spectral norm of a matrix. Any other norm will be indexed explicitly (e.g., $\|.\|_1,~\|.\|_{2, 2},~\|.\|_{2, \infty}$) to avoid confusion. Also for any vector $\beta$, we use $\beta_j$ to denote  its $j^{th}$ coordinate and  $\beta_{[g]}$ to denote the coordinates corresponding to the $g^{th}$ group.\\ \\
\textbf{Network Granger causal (NGC) estimates with group sparsity. }
Consider $n$  replicates from the NGC model~(\ref{eqn1:NGCdefn}), and denote the $n \times p$ observation matrix at time $t$ by $\mathcal{X}^t$. In econometric applications the data on $p$ economic variables across $n$ panels (firms, households etc.) can be observed over $T$ time points. For time course microarray data one typically observes the expression levels of $p$ genes across $n$ subjects over $T$ time points. After removing the panel specific fixed effects one assumes the common slope structure and independence across the panels. The data is high-dimensional if either $T$ or $p$ is large compared to $n$. In such a scenario, we assume the existence of an underlying group sparse structure, i.e.,  the support of each row of $A^{1:T} = \left[A^1: \cdots: A^T \right]$ in the model (\ref{eqn1:NGCdefn}) can be covered by a small number of groups $s$, where $s \ll (T-1)G$. Note that the groups can be misspecified in the sense that the coordinates of a group covering the support need not be all non-zero. Hence, for a properly specified group structure we shall expect $s \ll \| A^{1:T}_{i:} \|_0$. On the contrary, with many misspecified groups, $s$ can be of the same order, or even larger than $\| A^{1:T}_{i:} \|_0$.

Learning the true network of Granger causal effects $\{ (i, j) \in \{1, \ldots, p\}: A^t_{ij} \neq 0 \mbox{ for some } t \}$ is equivalent to recovering the correct sparsity pattern in $A^{1:(T-1)}$ and consistently estimating the non-zero effects $A^t_{ij}$. In the high-dimensional regression problems this is achieved by simulteneous regualrization and selection operators like lasso and group lasso. The group Granger causal estimates of the adjacency matrices $A^1, \ldots, A^{T-1}$  are obtained by solving the following optimization problem
\begin{equation}\label{eqn:NGCest_original}
\hat{A}^{1:T-1} = \displaystyle \argmin_{A^1,\cdots, A^{T-1}} \frac{1}{2n} \left \| \mathcal{X}^T - \displaystyle \sum_{t=1}^{T-1} \mathcal{X}^{T-t} \left( A^t \right)' \right \|_F^2 
			+ \lambda \displaystyle \sum_{t=1}^{T-1}\displaystyle \sum_{i=1}^p  \displaystyle \sum_{g=1}^G  w^{t}_{i, g}\| A^t_{i:[g]} \|  
\end{equation}
where $\mathcal{X}^t$ is the $n \times p$ observation matrix at time $t$, constructed by stacking $n$ replicates from the model (\ref{eqn1:NGCdefn}), $w^t$ is a $p \times G$ matrix of suitably chosen weights and $\lambda$ is a common regularization parameter. The optimization problem can be separated into the following $p$  penalized regression problems:
\begin{equation}\label{eqn:NGCest}
\hat{A}^{1:T-1}_{i:} =  \displaystyle \argmin_{\theta^1, \cdots, \theta^{T-1} \in \mathbb{R}^{p }}  \frac{1}{2n} \| \mathcal{X}^T_{:i} - \displaystyle \sum_{t=1}^{T-1} \mathcal{X}^{T-t} \theta^t \|^2 
			+ \lambda \displaystyle \sum_{t=1}^{T-1} \displaystyle \sum_{g=1}^G  w^{t}_{i, g}\| A^t_{i:[g]} \| , \ \ i=1,\cdots,p
\end{equation}
The order $d$ of the VAR model is estimated as $\hat{d} = \displaystyle \max_{1 \le t \le T-1} \{t: \hat{A}^t \neq \mathbf{0} \} $.

Different choices of weights $w^t_{i:g}$ lead to different variants of NGC estimates.  The regular NGC estimates correspond to the choices $w^t_{i, g} =1$ or $\sqrt{k_g}$, while for adaptive group NGC estimates the weights are chosen as $w^t_{i, g} =\left\| \hat{A}^t_{i:[g]} \right\|^{-1}$, where $\hat{A}^t$ are obtained from a regular NGC  estimation.

Thresholded NGC estimates are calculated by a two-stage procedure. The first stage involves a regular NGC estimation procedure. The second stage uses a bi-level thresholding strategy on the estimates $\hat{A}^t$. First, the estimated groups with $\ell_2$ norm less than a threshold ($\delta_{grp} = c \lambda, \, c>0$)  are set to zero. The second level of thresholding (within group) is applied if the \textit{a priori} available grouping information is not entirely reliable. $\hat{A}^t_{ij}$within an estimated group $\hat{A}^t_{i:[g]}$ is thresholded to zero if $\left|\hat{A}^t_{ij}\right| / \left\| \hat{A}^t_{i:[g]} \right\|$ is less than a threshold $\delta_{misspec} \in (0, 1)$. So, for every $t=1, \ldots, T-1$, if $j \in \mathcal{G}_g$, the thresholded NGC estimates are
\begin{equation*}
 \tilde{A}^t_{ij} = \hat{A}_{ij}^t I\left\{\left|\hat{A}^t_{ij}\right| \ge \delta_{misspec} \left\| \hat{A}^t_{i:[g]} \right\| \right\} I \left\{ \left\| \hat{A}^t_{i:[g]} \right\| \ge \delta_{grp} \right\} \label{thres_defn}
\end{equation*}
The tuning parameters $\lambda_{grp}$ and $\delta_{misspec}$ are chosen via cross-validation. The rationale behind this thresholding strategy is discussed in Section \ref{sec:selection}.


\section{Estimation Consistency of NGC estimates}\label{secresults}
The regular NGC estimates in (\ref{eqn:NGCest_original}) are obtained by solving $p$ separate group lasso programs with a common design matrix $\mathbf{X}_{n \times {p(T-1)}} = [ \mathcal{X}^1: \cdots: \mathcal{X}^{T-1}]$. This design matrix has $\bar{p} = (T-1)p$ columns which can be partitioned into $\bar{G} = (T-1)G$ groups $\left\{ \mathcal{G}_1, \ldots, \mathcal{G}_{\bar{G}}  \right\}$. We denote the sample Gram matrix by $C = \mathcal{X}' \mathcal{X} / n$. For the $i^{th}$ optimization problem, these $\bar{G}$ groups are penalized by $\lambda_g := \lambda \, w^t_{i,g}$.  Following \citet{lounici2011} the $\ell_2$ estimation error of the NGC estimates $\hat{A}^t$  can be shown to be bounded under certain restricted eigenvalue (RE) assumptions. 
These assumptions are common in the literature of high-dimensional regression \citep{lounici2011, bickel2009simultaneous, vandegeerconditions} and are known to be sufficient to guarantee consistent estimation of the regression coefficients even when the design matrix is singular. Of main interest, however, is to investigate the validity of these assumptions in the context of NGC models. This issue is addressed in Proposition \ref{spectralresult}. 

For $L > 0$, we say that a \textbf{Restricted Eigenvalue} (RE) assumption RE(s, L) is satisfied if there exists a positive number $\phi_{RE} = \phi_{RE}(s) > 0$ such that
\begin{equation}\label{RElounici}
\hspace{-0.1in}\min_{\begin{smallmatrix}J \subset \mathbb{N}_{\bar{G}},\, |J| \le s\\ \Delta \in \mathbb{R}^{\bar{p}} \backslash\{\mathbf{0}\} \end{smallmatrix}}  \left \{ \frac{\| \mathbf{X} \Delta\| }{\sqrt{n} \| \Delta_{[J]} \| } : \displaystyle \sum_{g \in J^c} \lambda_g \| \Delta_{[g]}\| \le L \displaystyle \sum_{g \in J} \lambda_g \| \Delta_{[g]} \| \right \} \ge \phi_{RE}
\end{equation}

The following theorem provides the convergence rate of the group NGC estimates under RE assumptions. The proof follows along the lines of   \citet{lounici2011} and is delegated to Appendix \ref{app_l2}.
\begin{thm}\label{thm:l2consistency}
Consider a regular NGC estimation problem \eqref{eqn:NGCest}. Let $\left\{ \mathcal{G}_{J^i}:~J^i \subseteq \mathbb{N}_{\bar{G}}\right\}$ denote the groups required to cover $support(A^{1:d}_{i:})$ and let $s = \max_{1 \le i \le p} \left| J^i\right|$. Suppose $\lambda, ~w^t_{i,g}$ in \eqref{eqn:NGCest_original} are chosen large enough such that  for some $\alpha > 0$,
\begin{eqnarray}
\lambda_g  \ge \frac{2 \sigma}{\sqrt{n}} \sqrt{\nn{C_{[g] [g]}}} \left(  \sqrt{k_g} + \frac{\pi}{\sqrt{2}} \sqrt{\alpha \, \log\,\bar{G}} \right)~~\mbox{for every $g \in \mathbb{N}_{\bar{G}}$,}
\end{eqnarray}
Also assume that the common design matrix $\mathbf{X} = [ \mathcal{X}^1: \cdots: \mathcal{X}^{T-1}]$ in the $p$ regression problems \eqref{eqn:NGCest}  satisfy $RE(2s, 3)$. 
Then, with probability at least $1 - 2 \bar{p} \bar{G}^{1-\alpha}$, 
\begin{equation}
\left\| \hat{A}^{1:\hat{d}} - A^{1:d} \right\|_F \le \frac{4\sqrt{10}}{\phi_{RE}^2 (2s)}  \, \frac{\max_{i} ~\sum_{g \in J^i} \lambda_g^2}{\lambda_{min} \, \sqrt{s }}
\end{equation}
\end{thm}

Note that group lasso achieves faster convergence rate (in terms of estimation and prediction error) than lasso if the groups are appropriately specified. For example, if all the groups are of equal size $k$ and $\lambda_g = \lambda$ for all $g$, then group lasso can achieve an $\ell_2$ estimation error of order $O\left(\sqrt{s}(\sqrt{k} + \sqrt{\log\,\bar{G}})/\sqrt{n}\right)$. In contrast, lasso's error is known to be of the order  $O\left(\sqrt{\| A^{1:d} \|_0 \, \log\,\bar{p}/n} \right)$, which establishes that group lasso has a lower error bound if $s \ll \| A^{1:d}_{i:} \|_0$. On the other hand, lasso will have a lower error bound if $s \asymp \| A^{1:d}_{i:} \|_0$, i.e., if the groups are highly misspecified.
\vspace{0.1in}

\textbf{Validity of RE assumption in Group NGC problems. }
In view of Theorem \ref{thm:l2consistency}, it is important to understand how stringent the RE condition is in the context of NGC problems. It is also important to find a lower bound on the RE coefficient $\phi_{RE}$, as it affects the convergence rate of the NGC estimates. For the panel-VAR setting, we can rigorously establish that the RE condition holds with overwhelming probability, as long as $n,~p$ grow at the same rate required for $\ell_2$-consistency. 

The following proposition achieves this objective in two steps. Note that each row of the design matrix $\mathbf{X}$ (common across the $p$ regressions) is independently distributed as $N(\mathbf{0}, \Sigma)$ where $\Sigma = Var(\mathbf{X}^{1:(T-1)})$. First, we exploit the spectral representation of the stationary VAR process to provide a lower bound on the minimum eigenvalue of $\Sigma$. An argument along the lines in \citet{raskutti2010REcorrgauss} then establishes that $\mathbf{X}$ satisfies the RE condition with overwhelming probability for sufficiently 
large $n$.
\begin{prop}\label{spectralresult}
(a)~~	Suppose the VAR(d) model of \eqref{eqn1:NGCdefn} is stable, stationary. Let $\Sigma = Var(\mathbf{X}^{1:(T-1)})$.  Then, $\Lambda_{min}(\Sigma) \ge m^2 := {\sigma^2}\left[ 1 + \frac{1}{2}(\mathbf{v}_{in} + \mathbf{v}_{out}) \right]^{-2}$, where $\mathbf{v}_{in}$ and $\mathbf{v}_{out}$ are the maximum incoming and outgoing effects at a node, i.e.,
\begin{equation*}
\mathbf{v}_{in} = \displaystyle \max _{1 \le i \le p} \displaystyle \sum_{t=1}^d \displaystyle \sum_{j=1}^p |A^t_{ij}| , ~~~~~~ \mathbf{v}_{out} = \displaystyle \max _{1 \le j \le p} \displaystyle \sum_{t=1}^d \displaystyle \sum_{i=1}^p |A^t_{ij}| 
\end{equation*}
(b)~~ In addition, suppose the replicates from the model \eqref{eqn1:NGCdefn} are i.i.d. Then, for any $s > 0$ there exist universal positive constants $c, c', c''$, such that if the sample size $n$ satisfies
\begin{equation*}
n > c'' \frac{16 \rho^2(\Sigma)}{m^2} \left( \frac{s (\sqrt{\log\,\bar{G}} + \sqrt{k_{max}})^2}{\lambda_{min}/\lambda_{max}} \right),
\mbox{~~ where ~} \rho^2(\Sigma) := \displaystyle \max_{1 \le g \le \bar{G}} \left\| \Sigma_{[g][g]} \right\|,
\end{equation*}
then $\mathbf{X}$ satisfies $RE(s, 3)$ with $\phi^2_{RE} \ge m^2/ 8$ with probability at least $1 - c' \exp(-cn)$.
\end{prop}
\textbf{Remark. } Proposition \ref{spectralresult} has two interesting consequences. First, it provides a lower bound on the RE constant $\phi_{RE}$ which is independent of $T$. So if the high dimensionality in the Granger causal network arises only from the time domain and not the cross-section ($T \rightarrow \infty,~ p,~G \mbox{~fixed}$), the stationarity of the VAR process guarantees that the rate of convergence depends only on the true order ($d$), and not $T$. Second, this result shows that the NGC estmates are consistent even if the node capacities $\mathbf{v}_{in}$ and $\mathbf{v}_{out}$ grow with $n,~p$ at an appropriate rate.

\section{Variable Selection Consistency of NGC estimates}\label{sec:selection}
In view of (\ref{eqn:NGCest}), to study the variable selection properties of NGC estimates it suffices to analyze the variable selection properties of $p$ generic group lasso estimates. The problem of group sparsity selection has been thoroughly investigated in the literature \citep{ adaptivegrplassofei, lounici2011}. The issue of selection and sign consistency within a group, however, is still unclear. Since group lasso does not impose sparsity within a group, all the group members are selected together \citep{grpbridge09} and it is not clear which ones are recovered with correct signs. This also leads to inconsistent variable selection if a group is misspecified. Several alternate penalized regression procedures have been proposed to overcome this shortcoming \citep{grpreg09,grpbridge09}. 
 The main idea behind these procedures is to somehow combine $\ell_2$ and $\ell_1$ norms in the penalty to encourage sparsity at both group and variable level. These estimators involve nonconvex optimization problems and are computationally expensive. Also their theoretical properties in a high dimensional regime are not well studied. 

We take a different approach to deal with the issue of group misspecification. Although group lasso penalty does not perform exact variable selection within groups, it performs regularization and shrinks the individual coefficients. We  utilize this regularization to detect misspecification within a group.
 To this end, we formulate a generalized notion of sign consistency, henceforth referred as ``direction consistency'', that provides insight into the properties of group lasso estimates within a single group. Subsequently, these properties are used to develop a   simple, easy to compute, 
thresholded variant of  group lasso which, in addition to group selection, achieves variable selection and sign consistency within groups.

Consider a generic group lasso regression problem of the linear model $y = X \beta^0 + \epsilon$ with $p$ variables partitioned into $G$ non-overlapping groups:
\begin{eqnarray}
\hat{\beta} =  \displaystyle \argmin_{\beta \in \mathbb{R}^p} \frac{1}{2n} \| \mathbf{Y} - \mathbf{X} \mathbf{\beta} \|^2 +  \displaystyle \sum_{g=1}^{\bar{G}} \lambda_g \|  \mathbf{\beta}_{[g]}\| \label{genericgrplasso} \\
\underbrace{\mathbf{\beta}^0}_{p \times 1} = [\underbrace{\mathbf{\beta}^0_{[1]}, \ldots, \mathbf{\beta}^0_{[s]}}_{k_1 + \ldots+k_s = q}, \underbrace{\mathbf{0}, \ldots, \mathbf{0}}_{p-q}] = [\mathbf{\beta}^0_{(1)}: \mathbf{\beta}^0_{(2)}]\\
\underbrace{\mathbf{X}}_{n \times p} = [\underbrace{\mathbf{X}_{(1)}}_{n \times q} : \underbrace{\mathbf{X}_{(2)}}_{n \times (p-q)}]  ~~~~~~~~C = \frac{1}{n} \mathbf{X}' \mathbf{X} = \left[ \begin{array} {cc}C_{11} & C_{12} \\ C_{21} & C_{22} \\	\end{array} \right]\label{signonsigpartn}
\end{eqnarray}


\textbf{Direction Consistency. }
For an $m$-dimensional vector $\mathbf{\tau} \in \mathbb{R}^m \backslash \{\mathbf{0} \}$ define its direction vector $D(\tau) = {\mathbf{\tau}}/{\| \mathbf{\tau} \|}$ , $D(\mathbf{0}) = \mathbf{0}$. In the context of a generic group lasso regression \eqref{signonsigpartn}, for a group $g \in S$ of size $k_g$, $D(\mathbf{\beta}^0_{[g]})$ indicates the direction of influence of $\mathbf{\beta}^0_{[g]}$ at a group level in the sense that it reflects the relative importance of the influential members within the group. Note that for $k_g=1$ the function $D(\cdot)$ simplifies to the usual $sgn(\cdot)$ function.\\ \\
\textbf{Definition. } An estimate $\mathbf{\hat{\beta}}$ of a generic group lasso problem \eqref{genericgrplasso} is \textbf{\textit{direction consistent}} at a rate $\delta_n$, if there exists a sequence of positive real numbers $\delta_n \rightarrow 0$ such that
\begin{equation}
\mathbb{P} \left(\| D(\mathbf{\hat{\beta}}_{[g]}) - D(\mathbf{\beta}^0_{[g]}) \| < \delta_n,~\forall g \in S, ~ \hat{\beta}_{[g]} = \mathbf{0}, ~\forall g \notin S \right) \rightarrow 1 \mbox{ as } n, p \rightarrow \infty.
\end{equation}
Now suppose $\mathbf{\hat{\beta}}$ is a direction consistent estimator.  Consider the set 
$\tilde{S}^n_{g} := \{ j \in \mathcal{G}_g:\,{| \mathbf{{\beta^0}}_j |}\, / \, {\|\mathbf{{\beta^0}}_{[g]} \|} > \delta_n \}$. $\tilde{S}^n_{g}$ can be viewed as a collection of influential group members within a group $\mathcal{G}_g$, which are ``detectable'' with a sample size of $n$. Then, it readily follows from the definition that 
\begin{equation}\label{sgnwithingrp}
\mathbb{P}(sgn(\mathbf{\hat{\beta}}_j) = sgn(\mathbf{\beta}_j), ~\forall j  \in \tilde{S}^n_g, \forall g \in \{ 1, \ldots, s\}) \rightarrow 1 \mbox{ as } n, p \rightarrow \infty.
\end{equation}
\textbf{Remark. } The latter observation connects the precision of group lasso estimates to the accuracy of \textit{a priori} available grouping information. In particular, if the pre-specified grouping structure is correct, i.e., all the members within a group have non-zero effect, then for a sufficiently large sample size we have $\tilde{S}_g^n = \mathcal{G}_g$ for all $g \in S$ and group lasso correctly estimates the sign of all the coordinates in the support. On the other hand, in case of a misspecified \textit{a priori} grouping structure (numerous zero coordinates in  $\mathbf{\beta}_g$ for $g \in S$), group lasso correctly estimates only the signs of the influential group members.\\
\textbf{Example. } We demonstrate the property of direction consistency using a small example. Consider a linear model with $8$ predictors
\begin{equation*}
y = 0.5 x_1 - 3 x_2 + 3 x_3+ x_4 -2 x_5 + 3 x_8 + e, ~~~ e \sim N(0, 1)
\end{equation*}
The coefficient vector $\beta^0$ is partitioned into four groups of size $2$, viz., $(0.5, -3), \, (3,1), \, (-2, 0)$ and $(0, 3)$. The last two groups are misspecified. We generated $n = 25$ samples from this model and ran group lasso regression with the above group structure. Figure \ref{dircontdemo} 
 shows the true coefficient vectors (solid) and their estimates (dashed) from five iterations of the above exercise. Note that even though the $\ell_2$ errors between  $\beta^0_{[g]}$ and $\hat{\beta}_{[g]}$ vary largely across the four groups, the distance between their projections on the unit circle, $\left \| D(\beta^0_{[g]}) - D(\hat{\beta}_{[g]})\right \|$, are comparatively stable across groups. In fact, Theorem \ref{selectconsist} shows that under certain irrepresentable conditions (IC) on the design matrix, it is possible to find a uniform (over all $g \in S$) upper bound $\delta_n$ on the $\ell_2$ gap of these direction vectors. This motivates a natural thresholding strategy to correct for the misspecification in groups (cf. Proposition \ref{propthres}). Even though a group $\beta^0_{[g]}$ is misspecified (i.e., lies on a coordinate axis), direction consistency ensures, with high probability, that the corresponding coordinate in $D(\hat{\beta}_{[g]})$ will be smaller than a threshold $\delta_n$ which is common across all groups in the support. 
\begin{figure}[t!]
\begin{center}
		\includegraphics[width = \textwidth]{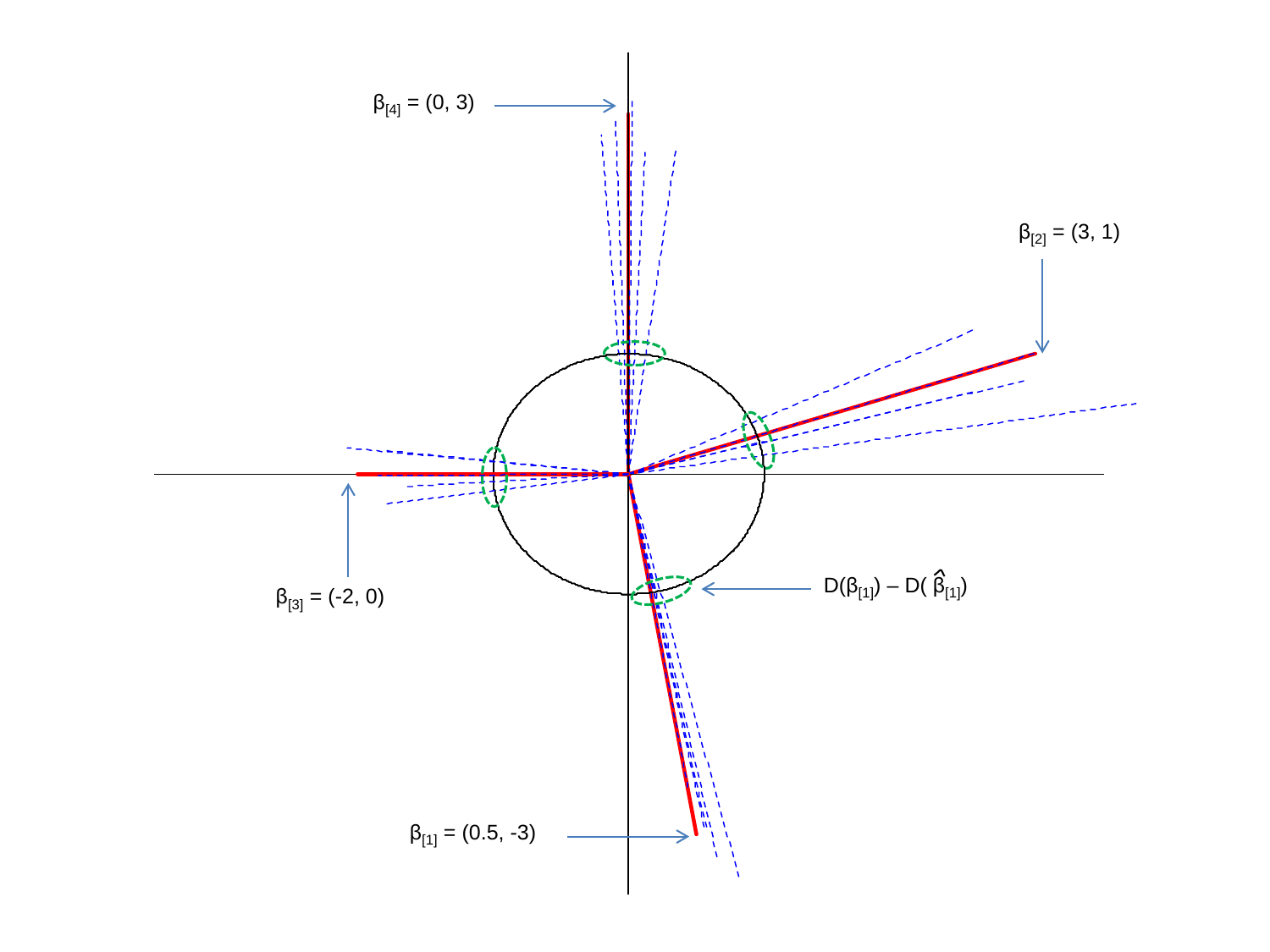}
	\caption{Example demonstrating direction consistency} 
\label{dircontdemo}
\end{center}
\end{figure}

\textbf{Group Irrepresentable Conditions (IC). }
Next, we define the IC required for direction consistency of group lasso estimates. 
Irrepresentable conditions are common in the literature of high-dimensional regression problems \citep{Zhaoyu06, vandegeerconditions} and are shown to be sufficient (and essentially necessary) for selection consistency of the lasso estimates. Further these conditions are known to be satisfied with high probability, if the population analogue of the Gram matrix belongs to the Toeplitz family \citep{Zhaoyu06, wainwright09}. In NGC estimation the population analogue of the Gram matrix $var(\mathbf{X}^{1:(T-1)})$ is block Toeplitz, so the irrepresentable assumptions are natural candidates for studying selection consistency of the estimates.
Consider the notations of \eqref{genericgrplasso} and \eqref{signonsigpartn}. Define $K = diag\left( \lambda_1 \mathbf{I}_{k_1}, \lambda_2 \mathbf{I}_{k_2}, \ldots, \lambda_s \mathbf{I}_{k_s} \right)$. \\
\textbf{Uniform Irrepresentable Condition (IC)} is satisfied if there exists $0 < \eta < 1$ such that for all $\tau \in \mathbb{R}^q $ with $\| \tau \|_{2, \infty} = \displaystyle \max_{1 \le g \le s} \| \tau_{[g]}\|_2  \le 1 $
\begin{equation}  \label{unifirrep}
\frac{1}{\lambda_g}\left \| \left[ C_{21} {\left(C_{11}\right)}^{-1}  K \tau  \right]_{[g]} \right \| < 1-\eta, ~ \forall g \notin S = \{ 1, \ldots, s\}%
\end{equation}
Note that the definition reverts to the usual IC for lasso when all groups correspond to singletons. 

The IC is more stringent than the RE condition and is rarely met if  the underlying model is not sparse. It can be shown that a slightly weaker version of this condition is necessary for direction consistency. We refer the readers to Appendix \ref{app_selection} for further discussion on the different irrepresentable assumptions and their properties.  Numerical evidence suggests that the group IC tends to be less stringent than the IC required for the selection consistency of lasso. We illustrate this using three small simulated examples.\\
\textit{Simulation $1$. } We constructed group sparse NGC models with $T = 5,\, p = 21, \, G = 7, k_g = 3$ and different levels network densities, where the network edges were selected at random and scaled so that $\| A^1 \| = 0.1$. For each of these models we generated $100$ samples of size  $n = 150$ and calculated the proportions of times the two types of irrepresentable conditions were met.  The results are dispayed in Figure \ref{irrep:sparsity}.\\
\textit{Simulation $2$. } We selected a VAR(1) model from the above class and drawn samples of size $n = 20, 50, \ldots, 250$. Figure \ref{irrep:n}  displays the proportions of times (based on $100$ simulations) the two ICs were met.  \\
\textit{Simulation $3$. } We generated $n = 200$ samples from the VAR(1) model of example $2$ for $T = 2, 3, 4, 5, 10, \dots, 40$. Figure \ref{irrep:T}  displays the proportions of times (based on $100$ simulations) the two ICs were met.  
\begin{figure}[t!]
\begin{center}
	\begin{subfigure}[b]{0.31\textwidth}
		\includegraphics[width = \textwidth]{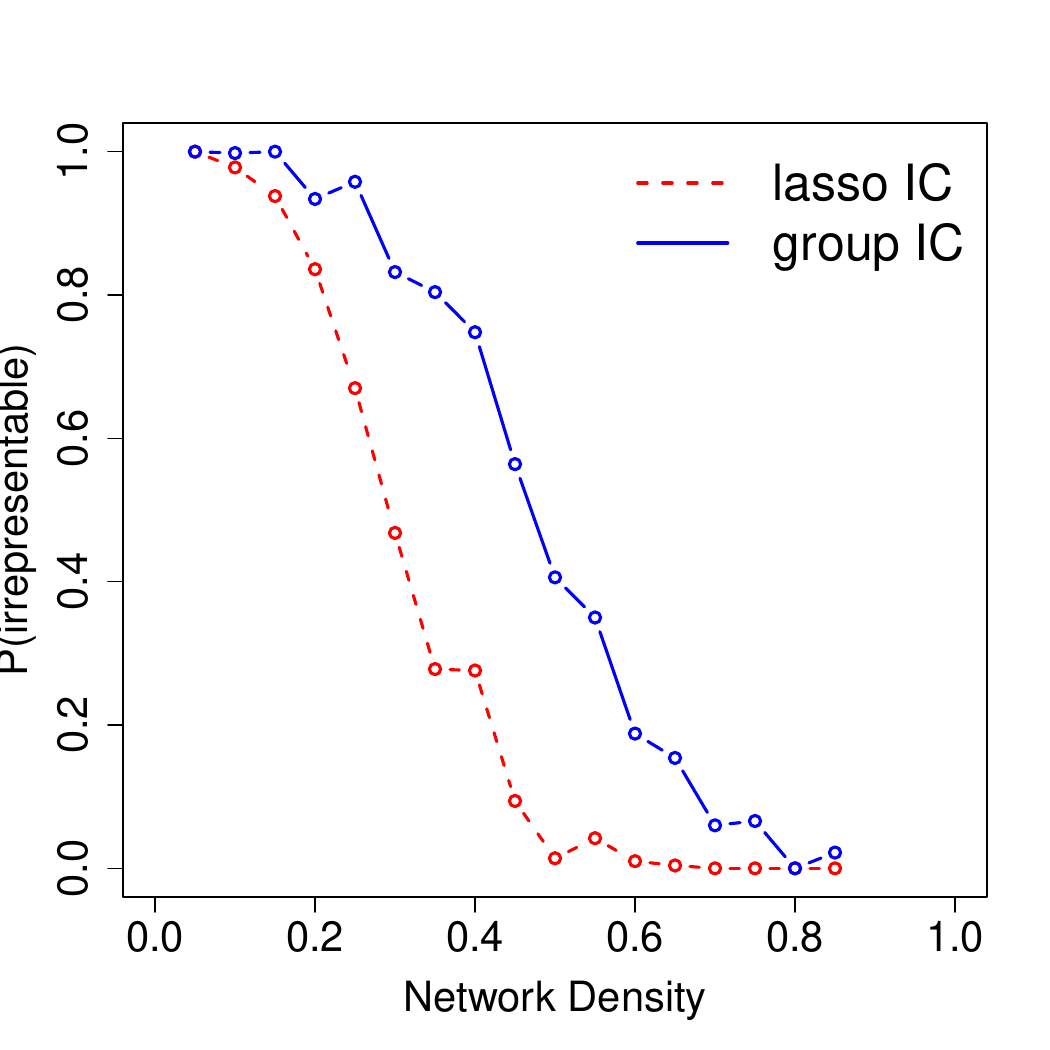}
		\caption{IC and Network Density}
		\label{irrep:sparsity}
	\end{subfigure}
	\begin{subfigure}[b]{0.31\textwidth}
		\includegraphics[width = \textwidth]{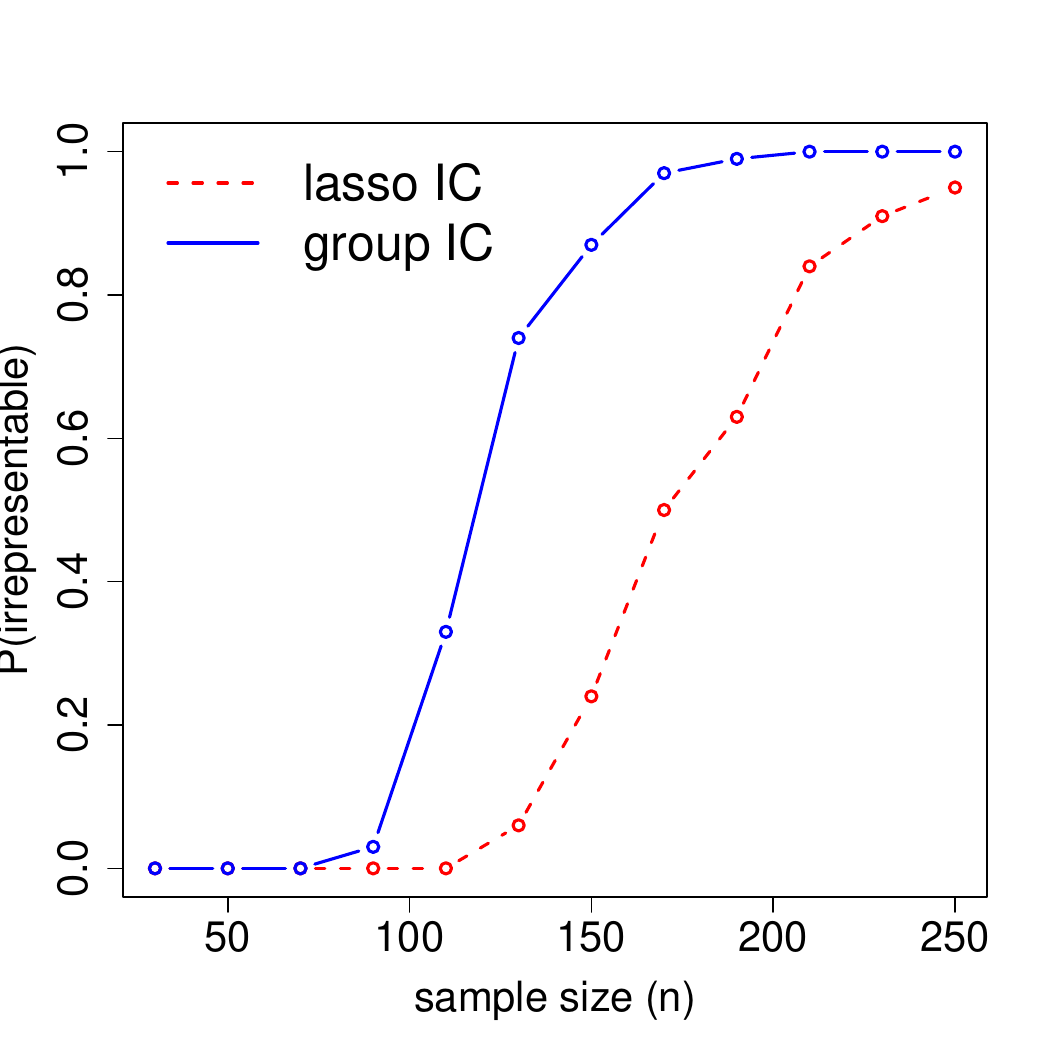}
		\caption{IC and Sample Size}
		\label{irrep:n}
	\end{subfigure}
	\begin{subfigure}[b]{0.31\textwidth}
		\includegraphics[width = \textwidth]{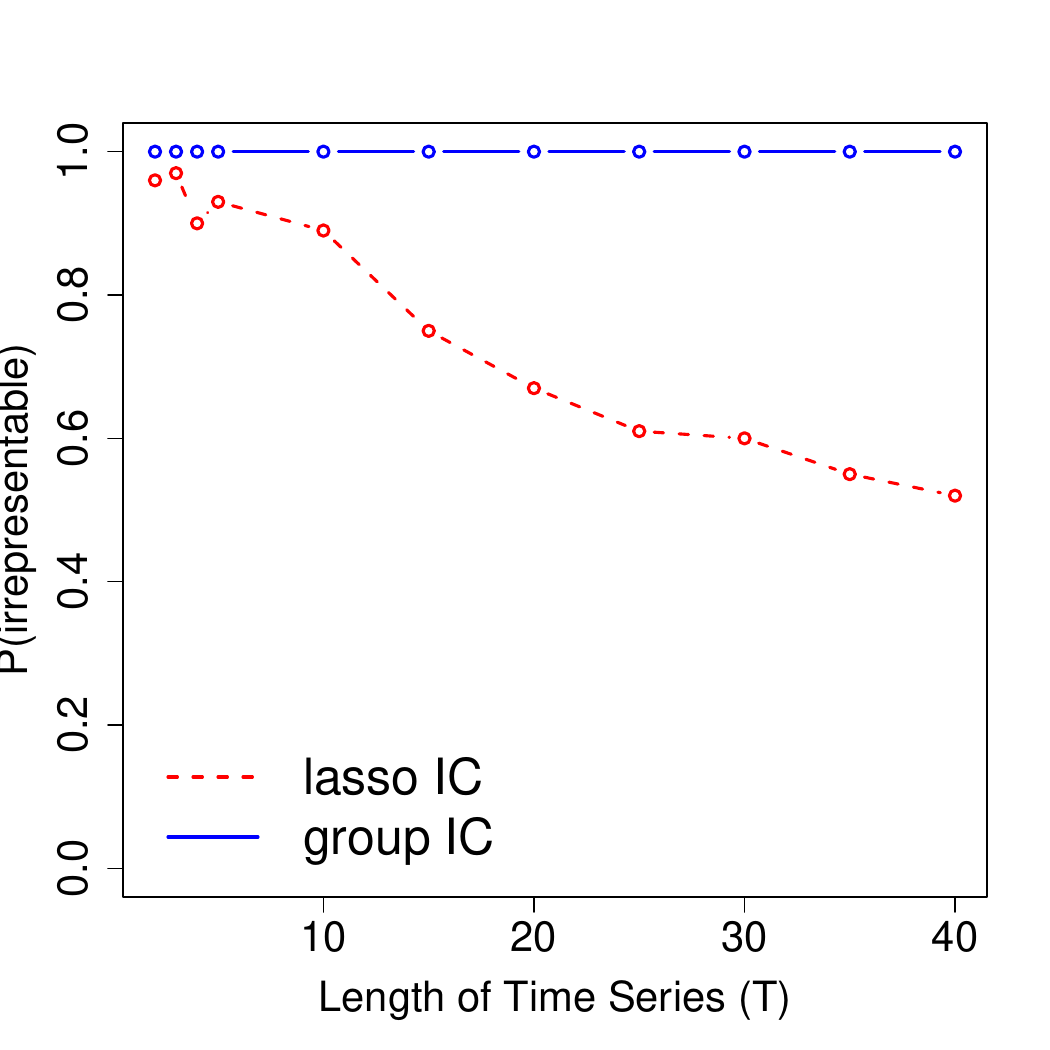}
		\caption{IC and Time Series length}
		\label{irrep:T} 
	\end{subfigure}
\caption{Comparison of lasso and group irrepresentable conditions in the context of group sparse NGC models. (a) group ICs tend to be met for dense networks where lasso IC fails to meet. (b) For the same network group IC is met with smaller sample size than required by lasso. (c) For longer time series group IC is satisfied more often than lasso IC.}
\label{irrep}
\end{center}
\end{figure}

\textbf{Selection consistency for generic group lasso estimates. }
For simplicity, we discuss the selection consistency properties of a generic group lasso regression problem with a common tuning parameter across groups,  i.e., $\lambda_g = \lambda$ for every $g \in \mathbb{N}_G$. Similar results can be obtained for more general choices of the tuning parameters.
\begin{thm}\label{selectconsist}
Assume that the group uniform IC holds with $1 - \eta$  for some $\eta > 0$. Then, for any choice of $\alpha > 0$,
\begin{eqnarray*}
\lambda &\ge& \displaystyle \max_{g \notin S} \frac{1}{\eta} \frac{\sigma}{\sqrt{n}} \sqrt{\left\| \left( C_{22}\right)_{[g][g]} \right\|} \left( \sqrt{k_g} + \frac{\pi}{\sqrt{2}} \sqrt{\alpha \, \log\,G} \right) ~~~\mbox{and}\\ 
\delta_n  &\ge&  \displaystyle \max_{g \in S} \frac{1}{\left\| \beta^0_{[g]} \right\|} 
\left( \lambda \sqrt{s} \left\| (C_{11})^{-1} \right\| + 
\frac{\sigma}{\sqrt{n}} \sqrt{\left\| (C_{11})^{-1}_{[g][g]} \right\| } 
\left(\sqrt{k_g} + \frac{\pi}{\sqrt{2}} \sqrt{\alpha \log\,G}\right) \right),
\end{eqnarray*}
 with probability greater than $1 -4 G^{1-\alpha}$, there exists a solution $\hat{\beta}$ satisfying
\begin{enumerate}
\item $\hat{\beta}_{[g]} = 0$ for all $g \notin S$,
\item $\left\| \hat{\beta}_{[g]} - \beta^0_{[g]} \right\| < \delta_n \left\| \beta_{[g]} \right\|$, and hence $\left\|  D(\hat{\beta}_{[g]}) - D(\beta^0_{[g]}) \right\|  < 2\delta_n\,$, for all $g \in S$. If $\delta_n < 1$, then $\hat{\beta}_{[g]} \neq 0$ for all $g \in S$.
\end{enumerate}
\end{thm}
\textbf{Remark. } The tuning parameter $\lambda$ can be chosen of the same order as required for $\ell_2$ consistency to achieve selection consistency within groups in the sense of \eqref{sgnwithingrp}.
Further, with the above choice of $\lambda$, $\delta_n$ can be chosen of the order of $O(\sqrt{s}(\sqrt{k_{max}} + \sqrt{\log\,G})/\sqrt{n})$. 
Thus, group lasso correctly identifies the group sparsity pattern and is direction consistent if $\sqrt{s}(\sqrt{k_{max}} + \sqrt{\log\,G})/\sqrt{n} \rightarrow 0$, the same scaling required for $\ell_2$ consistency.\\

\textbf{ Thresholding in Group NGC estimators. }
As described in Section \ref{secmodel}, regular group NGC estimates can be thresholded both at the group and coordinate levels. The first level of thresholding is motivated by the fact that lasso can select too many false positives  [cf. \citet{vandegeerthreshadaptejs2011}, \citet{zhou2010thresholded} and the references therein]. 
The second level of thresholding  employs the direction consistency of regular group NGC estimates to perform within group variable selection with high probability.  The following proposition demonstrates the benefit of these two types of thresholding. The second result is an immediate corollary of Theorem \ref{selectconsist}. Proof of the first result (thresholding at group level) requires some additional notations and is delegated to Appendix \ref{app_thres}.
\begin{thm}\label{propthres}
Consider a generic group lasso regression problem \eqref{genericgrplasso} with common tuning parameter $\lambda_g = \lambda$.\\
(i)~  Assume the RE(s, 3) condition of \eqref{RElounici} holds with a constant $\phi_{RE}$ and define
$\hat{\beta}^{thgrp}_{[g]} = \hat{\beta}_{[g]} \ind_{\| \hat{\beta}_{[g]} \| > 4 \lambda}$.
If $\hat{S} = \{ g \in \NN_G: \hat{\beta}^{thgrp}_{[g]} \ne \mathbf{0}\}$, then $|\hat{S} \backslash S| \le \frac{ s}{\phi_{RE}^2/12}$, with probability at least $1 - 2G^{1-\alpha}$.\\
(ii)~Assume that uniform IC holds with $1-\eta$ for some $\eta > 0$. Choose $\lambda$ and $\delta_n$ as in  Theorem~\ref{selectconsist}  and define
\begin{equation*}
\hat{\beta}^{thgrp}_{j} = \hat{\beta}_j \ind \{|\hat{\beta}_j|/ \|\hat{\beta}_{[g]} \| > 2 \, \delta_n \} \mbox{  for all $j \in \mathcal{G}_g$ }
\end{equation*}
Then $sgn(\beta^0_j) = sgn(\hat{\beta}^{thgrp}_j)~\forall \, j \in \mathbb{N}_p$ with probability at least  $1-4G^{1-\alpha}$, if $\displaystyle \min_{j \in supp(\beta^0)} |\beta^0_j|  > 2 \delta_n \, \| \beta^0_{[g]} \|$ for all $j \in \calG_g$, i.e., the effect of every non-zero member in a group is ``visible" relative to the total effect from the group.
\end{thm}


\section{Performance Evaluation}\label{secsim}
\begin{figure}[t!]
\includegraphics[width = \textwidth, height = 0.15\textwidth, trim = 0in 0.1in 0in 0.5in]{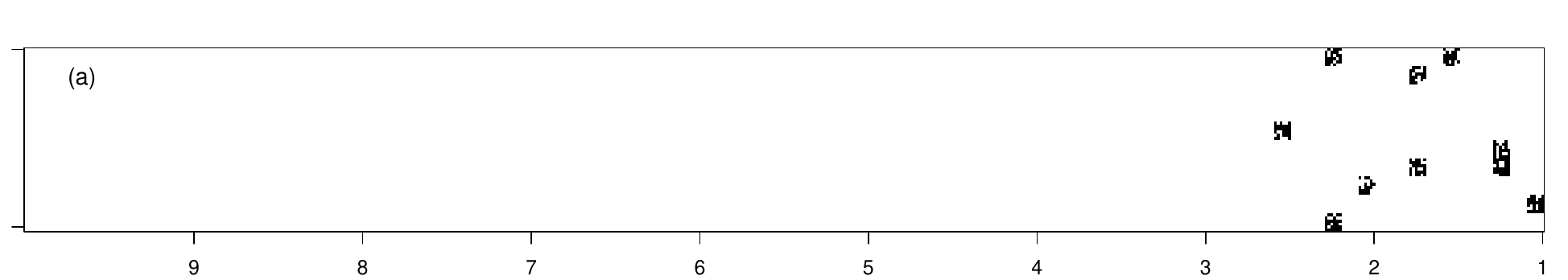}
\includegraphics[width = \textwidth, height = 0.15\textwidth, trim = 0in 0.1in 0in 0.5in]{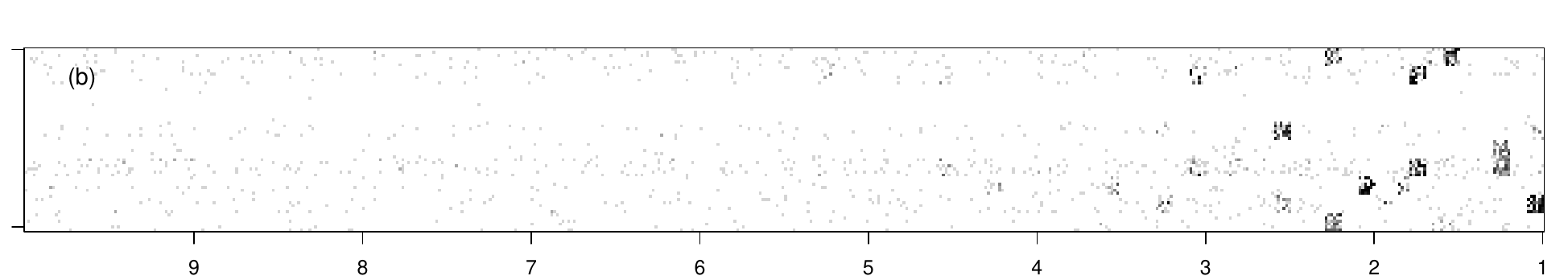}
\includegraphics[width = \textwidth, height = 0.15\textwidth, trim = 0in 0.1in 0in 0.5in]{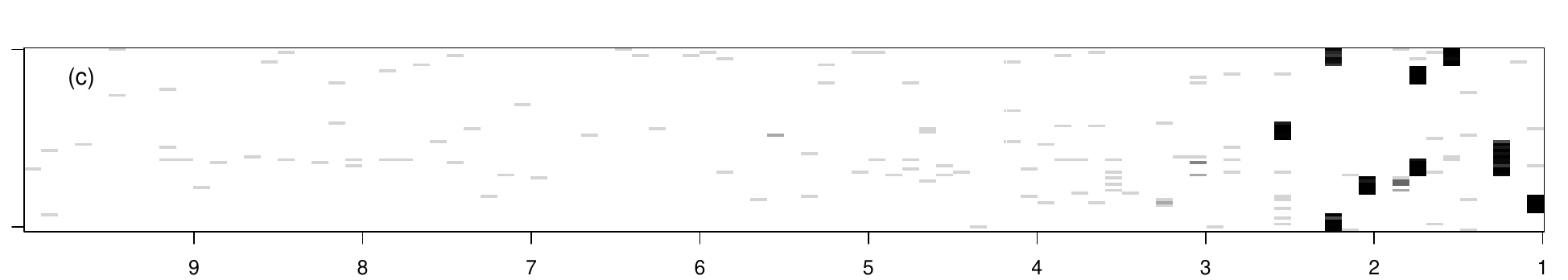}
\includegraphics[width = \textwidth, height = 0.15\textwidth, trim = 0in 0.1in 0in 0.5in]{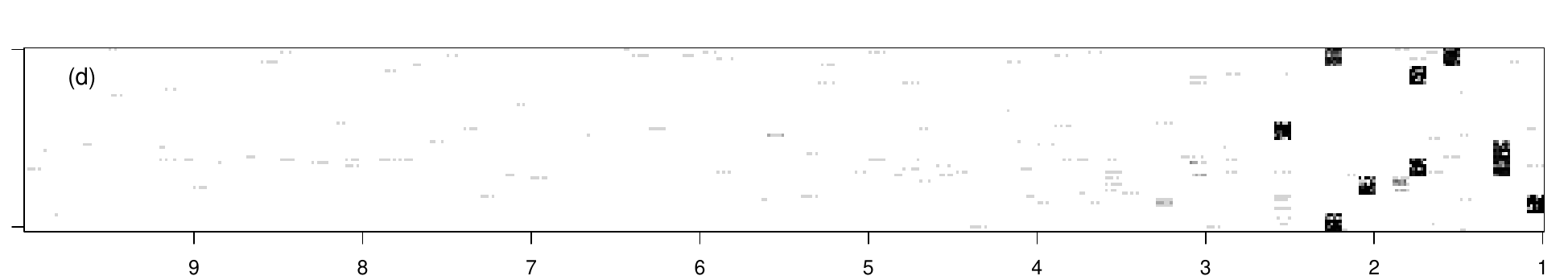}
\caption{Estimated adjacency matrices of a misspecified NGC model with p = 60, T = 10, n = 60: (a) True, (b) Lasso, (c) Group Lasso, (d) Thresholded Group Lasso. The grayscale represents the proportion of times an edge was detected in $100$ simulations.}
\label{fig_adj}
\end{figure}
We evaluate the performances of regular, adaptive and thresholded variants of the group NGC estimators through an extensive
simulation study, and compare the results to those obtained from lasso estimates. The R package \texttt{grpreg} \citep{grpreg09} was
used to obtain the group lasso estimates.
The settings considered are:\\
(a) \textit{Balanced groups of equal size}: i.i.d samples of size $n=60,110,160$ are generated from lag-2 ($d=2$)
VAR models on $T=5$ time points, comprising of $p=60,120,200$ nodes partitioned into groups of equal size in the range 3-5.\\
(b) \textit{Unbalanced groups:} We retain the same setting as before, however the corresponding node set is partitioned into one larger group of size 10 and many groups of size 5.\\
(c) \textit{Misspecified balanced groups:} i.i.d samples of size $n=60,110,160$ are generated from lag-2 ($d=2$)
VAR models on $T=10$ time points, comprising of $p=60,120$ nodes partitioned into groups of size 6. Further, for each group there is a 30\% misspecification rate, namely that for every parent group of a downstream node, 30\% of the group members do not exert any effect on it.

The tuning parameter $\lambda$ is chosen based on a grid search in the interval
$[C_1 \lambda_e, C_2 \lambda_e]$ where $\lambda_e  = \sqrt{2 \,\log\,p/n}$ for lasso and $\sqrt{2 \,\log\,G/n}$ for group lasso, using
a $19:1$ sample-splitting. The thresholding parameters are selected as $ \delta_{grp} = 0.7 \lambda \sigma$ at the group level and $\delta_{misspec} = n^{-0.2}$ within groups. Finally, within group thresholding is applied only when the group structure is misspecified.

The following performance metrics were used for comparison purposes: (i) $Precision= TP/(TP+FP)$ , (ii) $Recall = TP/(TP+FN)$ and (iii) Matthew's Correlation coefficient (MCC) defined as
\begin{equation*}
\frac{(TP \times TN) - (FP \times FN)}{((TP + FP)\times (TP + FN) \times(TN + FP) \times(TN + FN))^{1/2}}
\end{equation*}
 where $TP$, $TN$, $FP$ and $FN$ correspond to true positives, true negatives, false positives and false
negatives in the estimated network, respectively.

\begin{table}[h!]
\centering
\caption{Performance of different regularization methods in estimating graphical Granger causality with {\textbf{balanced}} group sizes and no misspecification; $d=2$, $T=5$, $SNR=1.8$. Precision ($P$), Recall ($R$), MCC are given in percentages (numbers in parentheses give standard deviations). ERR LAG gives the error associated with incorrect estimation of VAR order.}
\label{tbl:bal}
{\scriptsize
\begin{tabular}{r|l|ccc|ccc|ccc|}   
\multicolumn{2}{c}{} & \multicolumn{3}{c}{$p=60$, $|E|=351$} & \multicolumn{3}{c}{$p=120$, $|E|=1404$} & \multicolumn{3}{c}{$p=200$, $|E|=3900$} \\
\multicolumn{2}{c}{} & \multicolumn{3}{c}{Group Size=3} & \multicolumn{3}{c}{Group Size=3} & \multicolumn{3}{c}{Group Size=5} \\ \cline{2-11}
\multicolumn{1}{c|}{} &	n & 160 & 110 & 60 & 160 & 110 &	60	    &	160	    &	110	    &	60	\\ \cline{2-11}
P &	Lasso & 80(2)	&	75(2)	&	66(4)	&	69(1)	&	62(2)	&	52(2)	&	52(1)	&	47(1)	&	38(1)	\\
 &	Grp	  &  95(2)	&	91(4)	&	83(7)	&	91(3)	&	80(5)	&	68(7)	&	78(4)	&	72(3)	&	59(6)	\\
 &	Thgrp & 96(1)	&	92(3)	&	86(6)	&	93(3)	&	83(5)	&	70(7)	&	82(4)	&	76(3)	&	64(6)	\\
 &	Agrp  & 96(2)	&	92(4)	&	83(7)	&	92(3)	&	82(5)	&	69(7)	&	81(3)	&	74(3)	&	60(6)	\\ \cline{2-11}
R &	Lasso & 71(2)	&	54(2)	&	31(2)	&	54(1)	&	40(1)	&	22(1)	&	38(1)	&	28(1)	&	15(1)	\\
 &	Grp	  &  99(1)	&	93(3)	&	71(7)	&	91(2)	&	81(2)	&	48(8)	&	84(1)	&	70(2)	&	41(4)	\\
 &	Thgrp & 99(1)	&	93(3)	&	71(7)	&	91(2)	&	81(2)	&	48(8)	&	84(2)	&	69(2)	&	41(3)	\\
 &	Agrp  & 99(1)	&	93(3)	&	71(7)	&	91(2)	&	81(2)	&	47(8)	&	84(1)	&	69(2)	&	40(4)	\\ \cline{2-11}
MCC	& Lasso & 75(2)	&	63(2)	&	45(3)	&	60(1)	&	49(1)	&	33(1)	&	43(1)	&	35(1)	&	23(1)	\\
 &	Grp	  &  97(1)	&	92(3)	&	76(5)	&	91(1)	&	80(2)	&	56(2)	&	81(2)	&	70(2)	&	48(2)	\\
 &	Thgrp & 98(1)	&	93(2)	&	78(5)	&	92(1)	&	81(2)	&	57(3)	&	83(2)	&	72(2)	&	50(3)	\\
 &	Agrp  & 97(1)	&	92(3)	&	76(5)	&	91(1)	&	81(2)	&	56(3)	&	82(2)	&	71(2)	&	48(2)	\\ \cline{2-11} 
ERR & Lasso	&	10.5	&	11.3	&	13.9	&	16.63	&	17.37	&	16.69	&	19.79	&	20	&	18.52	\\
LAG & Grp	&	3.19	&	6.95	&	12.76	&	4.86	&	10.77	&	12.65	&	4.21	&	5.27	&	7.8	\\
& Thgrp	&	2.83	&	5.87	&	10.01	&	3.98	&	9.03	&	11.19	&	3.06	&	3.91	&	5.68	\\
& Agrp	&	3.13	&	6.89	&	12.59	&	4.63	&	10.37	&	12.34	&	3.58	&	4.87	&	7.59	\\
 \cline{2-11}
\end{tabular}
}
\end{table}

The results for the balanced settings are given in Table~\ref{tbl:bal}. The average and standard deviations (in parentheses) of the performance metrics are presented for each setup. The Recall for $p=60$ shows that even for a network with
$60 \times (5-1) = 240$ nodes and $|E| = 351$ true edges, the group NGC estimators recover about $71\%$ of the true edges with a sample size as low as $n = 60$, while lasso based NGC estimates recover only $31\%$ of the true edges. The three group NGC estimates have comparable performances in all the cases. However thresholded lasso shows slightly higher precision than the other group NGC variants for smaller sample sizes (e.g., $n=60, p = 200$). The results for $p = 60, n = 110$ also display that lower precision of lasso is caused partially by its inability to estimate the order of the VAR model correctly, as measured by ERR LAG=Number of falsely connected edges from lags beyond the true order of the VAR model divided by the number of edges in the network ($|E|$).
This finding is nicely illustrated in Figure~\ref{fig_adj} and Table~\ref{tbl:bal}. The group penalty encourages edges from the nodes of the same group to be picked up together. Since the nodes of the same group are also from the same time lag, the group variants have substantially lower ERR LAG. For example, average ERR LAG of lasso for $p=200, ~n=160$ is $19.79\%$ while the average ERR LAGs for the group lasso variants are in the range $3.06\% - 4.21\%$.

\begin{table}[h!]
\centering
\caption{Performance of different regularization methods in estimating graphical Granger causality with {\textbf{unbalanced}} group sizes and no misspecification; $d=2$, $T=5$, $SNR=1.8$. Precision ($P$), Recall ($R$), MCC are given in percentages (numbers in parentheses give standard deviations). ERR LAG gives the error associated with incorrect estimation of VAR order.}
\label{tbl:unbal}
{\scriptsize
\begin{tabular}{r|l|ccc|ccc|ccc|}   
\multicolumn{2}{c}{} & \multicolumn{3}{c}{$p=60$, $|E|=450$} & \multicolumn{3}{c}{$p=120$, $|E|=1575$} & \multicolumn{3}{c}{$p=200$, $|E|=4150$} \\
\multicolumn{2}{c}{} & \multicolumn{3}{c}{Groups=$1\times10, 10\times5$} & \multicolumn{3}{c}{Groups=$1\times10, 22\times5$} & \multicolumn{3}{c}{Groups=$1\times10, 38\times5$} \\ \cline{2-11}
\multicolumn{1}{c|}{} &	n & 160 & 110 &	60 & 160 & 110 & 60	    &	160	    &	110	    &	60	\\ \cline{2-11}
P	&	Lasso & 72(2)	&	69(3)	&	62(2)	&	51(1)	&	48(1)	&	41(1)	&	61(1)	&	53(1)	&	42(2)	\\
	&	Grp   & 84(4)	&	79(6)	&	76(9)	&	55(5)	&	47(5)	&	40(6)	&	86(3)	&	77(5)	&	66(7)	\\
	&	Thgrp & 86(4)	&	82(7)	&	78(11)	&	60(6)	&	50(7)	&	40(5)	&	88(2)	&	79(6)	&	69(6)	\\
	&	Agrp  & 85(3)	&	81(5)	&	77(9)	&	59(5)	&	51(5)	&	42(6)	&	88(2)	&	78(5)	&	67(6)	\\ \cline{2-11}
R	&	Lasso & 45(2)	&	35(2)	&	22(2)	&	43(1)	&	34(1)	&	22(1)	&	23(1)	&	15(0)	&	7(0)	\\
	&	Grp   & 94(3)	&	87(5)	&	61(8)	&	88(2)	&	75(5)	&	48(6)	&	73(3)	&	49(6)	&	22(5)	\\
	&	Thgrp & 95(2)	&	88(4)	&	62(8)	&	89(3)	&	77(4)	&	50(5)	&	73(3)	&	50(6)	&	21(5)	\\
	&	Agrp  & 94(3)	&	87(5)	&	61(8)	&	88(2)	&	75(5)	&	48(6)	&	73(3)	&	49(6)	&	22(5)	\\ \cline{2-11}
MCC	&	Lasso & 56(2)	&	48(2)	&	35(2)	&	46(1)	&	39(1)	&	29(1)	&	36(1)	&	28(1)	&	17(1)	\\
	&	Grp   & 89(3)	&	82(4)	&	67(5)	&	68(3)	&	58(3)	&	42(3)	&	79(1)	&	61(3)	&	37(3)	\\
	&	Thgrp & 90(3)	&	84(4)	&	68(6)	&	72(4)	&	61(4)	&	43(2)	&	80(1)	&	62(3)	&	37(3)	\\
	&	Agrp  & 89(3)	&	83(4)	&	67(6)	&	71(3)	&	60(3)	&	43(3)	&	79(1)	&	61(3)	&	37(3)	\\ \cline{2-11}
ERR & Lasso	&	10.59	&	10.74	&	11.76	&	18.3	&	18.72	&	18.76	&	11.54	&	10.93	&	9.29	\\
LAG	& Grp	&	7.04	&	9.85	&	13.04	&	12.53	&	14.71	&	13.06	&	4.8	&	6.41	&	6.85	\\
	& Thgrp	&	6.58	&	8.98	&	11.1	&	9.6	&	11.9	&	10.9	&	4.06	&	5.65	&	5.7	\\
	& Agrp	&	6.74	&	9.19	&	12.96	&	10.81	&	12.78	&	11.79	&	4.55	&	6.2	&	6.81	\\
\cline{2-11}
\end{tabular}
}
\end{table}

The results for the unbalanced networks are given in Table \ref{tbl:unbal}. As in the balanced group setup, in almost all the simulation settings the group NGC variants  outperform the lasso estimates with respect to all three performance metrics. However the performances of the different variants of group NGC are comparable and tend to have higher standard deviations than the lasso estimates. Also the average ERR LAGs for the group NGC variants are substantially lower than the average ERR LAG for lasso demonstrating the advantage of group penalty.  Although the conclusions regarding the comparisons of lasso and group NGC estimates remain unchanged it is evident that the performances of all the estimators are affected by the presence of one large group, skewing the uniform nature of the network. For example the MCC measures of group NGC estimates in a  balanced network with $p=60$ and $|E| = 351$ vary around $97 - 98\%$ which lowers to $89\% - 90\%$ when the groups are unbalanced.

\begin{table}[h!]
\centering
\caption{Performance of different regularization methods in estimating graphical Granger causality with {\textbf{misspecified}} groups (30\% misspecification); $d=2$, $T=10$, $SNR=2$. Precision ($P$), Recall ($R$), MCC are given in percentages (numbers in parentheses give standard deviations). ERR LAG gives the error associated with incorrect estimation of VAR order.}
\label{tbl:mis}
{\scriptsize
\begin{tabular}{r|l|ccc|ccc|}   
\multicolumn{2}{c}{} & \multicolumn{3}{c}{$p=60$, $|E|=246$} & \multicolumn{3}{c}{$p=120$, $|E|=968$} \\
\multicolumn{2}{c}{} & \multicolumn{3}{c}{Group Size=6} & \multicolumn{3}{c}{Group Size=6} \\ \cline{2-8}
\multicolumn{1}{c|}{} &	n & 160 & 110 & 60 & 160 & 110 &	60	    \\ \cline{2-8}
P	&	Lasso	&	88(2)	&	85(3)	&	77(5)	&	59(1)	&	55(1)	&	49(2)	\\
	&	Grp	    &	65(2)	&	66(2)	&	66(3)	&	43(3)	&	44(4)	&	38(4)	\\
	&	Thgrp	&	87(3)	&	88(3)	&	85(3)	&	56(6)	&	56(6)	&	51(7)	\\
	&	Agrp	&	65(2)	&	66(2)	&	66(3)	&	45(2)	&	45(4)	&	39(4)	\\ \cline{2-8}
R	&	Lasso	&	80(3)	&	63(3)	&	37(2)	&	66(1)	&	54(1)	&	35(1)	\\
	&	Grp	    &	100(0)	&	98(2)	&	82(6)	&	87(2)	&	78(3)	&	59(4)	\\
	&	Thgrp	&	100(0)	&	98(2)	&	79(6)	&	86(2)	&	79(3)	&	57(4)	\\
	&	Agrp	&	100(0)	&	98(2)	&	82(6)	&	86(2)	&	78(3)	&	58(3)	\\ \cline{2-8}
MCC	&	Lasso	&	84(2)	&	73(2)	&	53(3)	&	62(1)	&	54(1)	&	41(1)	\\
	&	Grp	    &	81(1)	&	80(2)	&	74(4)	&	61(2)	&	58(3)	&	47(2)	\\
	&	Thgrp	&	93(2)	&	93(2)	&	82(4)	&	69(4)	&	66(4)	&	53(3)	\\
	&	Agrp	&	81(1)	&	80(2)	&	74(4)	&	62(2)	&	59(2)	&	47(2)	\\ \cline{2-8}
ERR	&	Lasso	&	12.63	&	17.05	&	22.41	&	45.09	&	49.68	&	53.4	\\ 
LAG	&	Grp	    &	9.43	&	8.78	&	15.12	&	18.22	&	18.43	&	29.26	\\
	&	Thgrp	&	6.45	&	5.34	&	8.02	&	11.81	&	12.84	&	15.57	\\
	&	Agrp	&	9.11	&	8.78	&	14.96	&	16.32	&	16.9	&	27.69	\\ \cline{2-8}
\end{tabular}
}
\end{table}

The results for misspecified groups are given in Table \ref{tbl:mis}. Note that for higher sample size $n$, the MCC of lasso and regular group lasso are comparable. However, the thresholded version of group lasso ($\delta_{misspec} = n^{-0.2}$ used for within group selection) achieves significantly higher MCC than the rest. This demonstrates the advantage of using the directional consistency of group lasso estimators to perform within group variable selection. We would like to mention here that a careful choice of the thresholding parameters $\delta_{grp}$ and $\delta_{misspec}$ via cross-validation or other model selection criteria indicate improvement in the performance of thresholded group lasso; however, we do not pursue these methods here as they require grid search over many tuning parameters or an efficient estimator of the degree of freedom of group lasso.

In summary, the results clearly show that all variants of group lasso NGC outperform the lasso-based ones, whenever the grouping structure of the variables is known and correctly specified. Further, their performance depends on the composition of group sizes. On the other hand, if the a priori known group structure is moderately misspecified lasso estimates produce comparable results to regular and adaptive group NGC ones, while thresholded group estimates outperform all other methods, as expected.

\section{Application}\label{secdata}
\textbf{Example: T-cell activation. }
Estimation of gene regulatory networks from expression data is a fundamental problem in functional genomics \citep{friedman2004inferring}. 
Time course data coupled with NGC models are informationally rich enough for the task at hand. The data for this application come
from \citet{rangel2004modeling}, where expression patterns of genes involved in T-cell activation were studied with the goal of discovering regulatory mechanisms that govern them in response to external stimuli. Activated T-cells are involved in regulation of effector cells (e.g. B-cells) and play a central role in mediating immune response. The available data comprising of $n=44$ samples of $p=58$ genes, measure the
cells response at 10 time points, $t = 0, 2, 4, 6, 8, 18, 24, 32, 48, 72$ hours after their stimulation with a T-cell receptor independent activation mechanism.
We concentrate on data from the first 5 time points, that correspond to early response mechanisms in the cells.

Genes are often grouped based on their function and activity patterns into biological pathways. Thus, the knowledge of gene functions and their membership in biological pathways can be used as inherent grouping structures in the proposed group lasso estimates of NGC. Towards this, we used available biological knowledge to define groups of genes based on their biological function. Reliable information for biological functions were found from the literature for 38 genes, which were retained for further analysis. These 38 genes were grouped into 13 groups with the number of genes in different groups ranging from 1 to 5.

\begin{figure}[t!]
\centering
{\includegraphics[scale = 0.5, trim = 0in 0.9in 0in 0in]{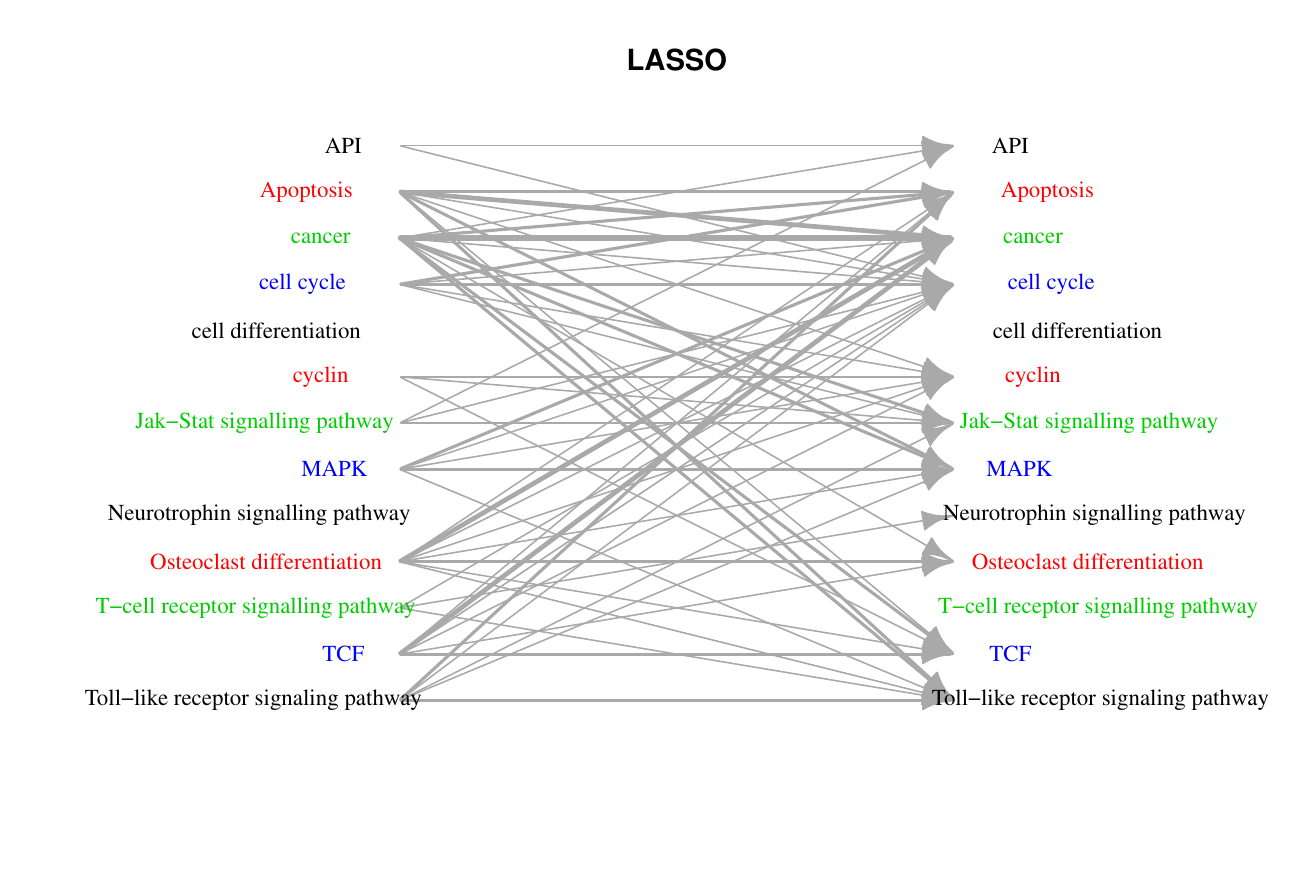}} 
{\includegraphics[scale = 0.5, trim = 0in 0.4in 0in 0in]{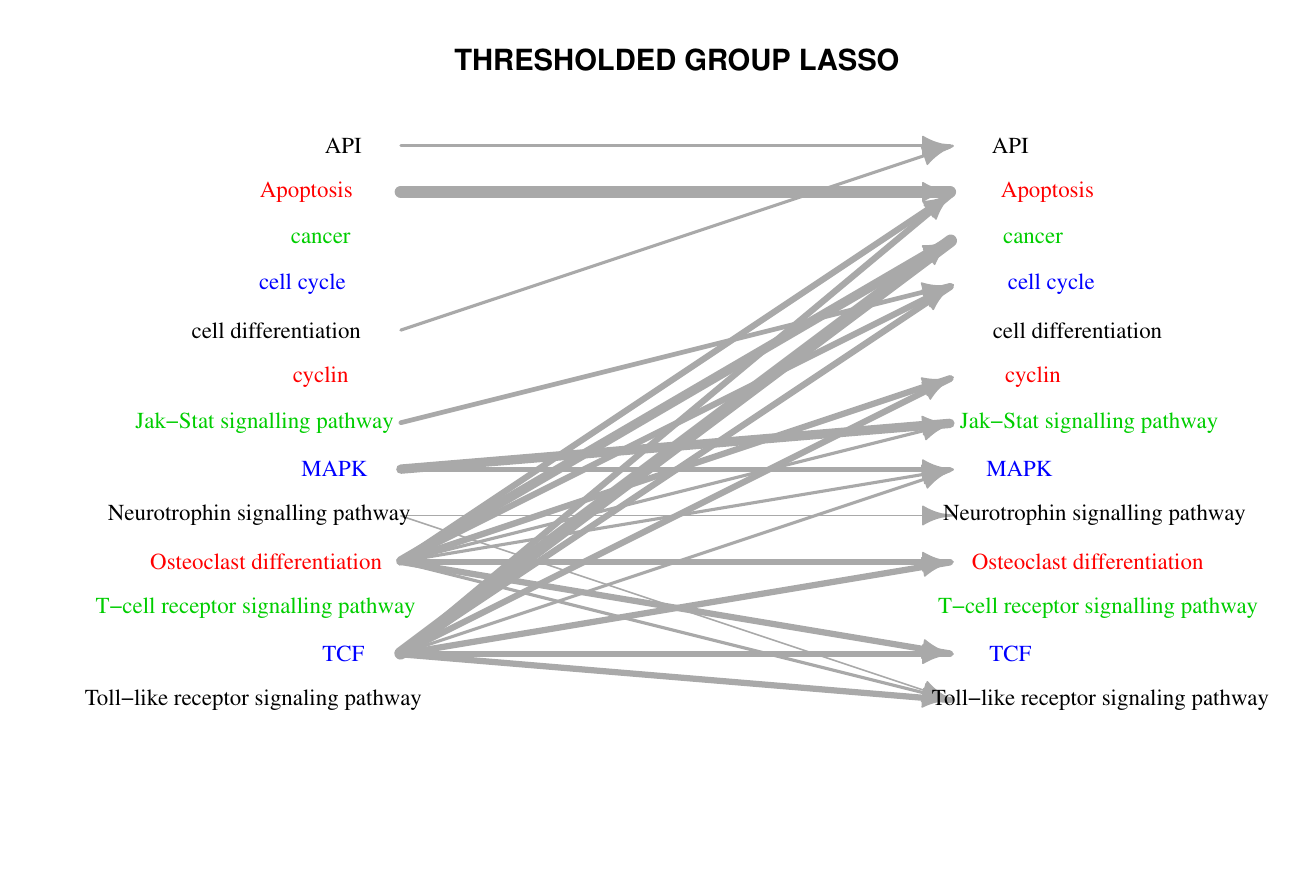}}
\caption{Estimated Gene Regulatory Networks of T-cell activation. Width of edges represent the number of effects between two groups, and the network represents the aggregated regulatory network over 3 time points.}
\label{fig_gene}
\end{figure}


\begin{table}[t!]
  \caption{Mean and standard deviation of MSE for different NGC estimates}\label{tbl_gene}
  \centering
  \begin{tabular}{|c|c c c c|} \hline
  \,	& Lasso	& Grp & Agrp & Thgrp \\ \hline
  mean & 0.649 & 0.456 & 0.457 & 0.456 \\
  stdev & 0.340 & 0.252 & 0.251 & 0.252 \\ \hline
  \end{tabular}
\end{table}

Figure~\ref{fig_gene} shows the estimated networks based on lasso and thresholded group lasso estimates, where for ease of representation the nodes of the network correspond to groups of genes.
In this case, estimates from variants of group NGC estimator were all similar, and included a number of known regulatory mechanisms in T-cell activation, not present in the regular lasso estimate. For instance, \citet{waterman1990purification} suggest that TCF plays a significant role in activation of T-cells, which may describe the dominant role of this group of genes in the activation mechanism. On the other hand, \citet{kim2005nuclear} suggest that activated T-cells exhibit high levels of osteoclast-associated receptor activity which may attribute the large number of associations between member of osteoclast differentiation and other groups. Finally, the estimated networks based on variants of group lasso estimator also offer improved estimation accuracy in terms of mean squared error (MSE) despite having having comparable complexities to their regular lasso counterpart (Table~\ref{tbl_gene}), which further confirms the findings of other numerical studies in that paper.

\textbf{Example: Banking balance sheets application. }
\begin{figure}[t!]
\centering
{\includegraphics[scale = 0.5, trim= 0in 1in 0.5in 0.5in, clip=true ]{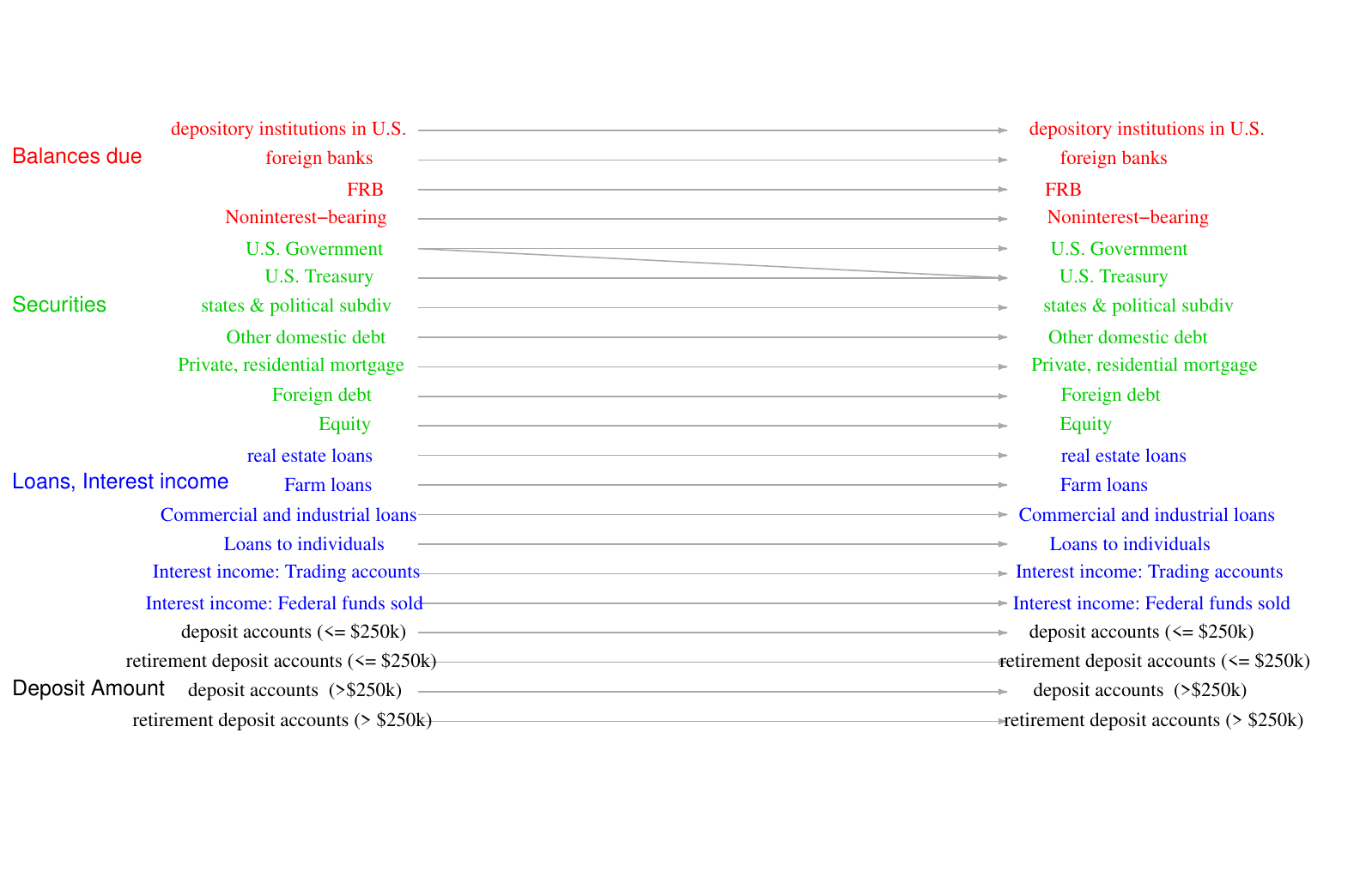}}
{\includegraphics[scale = 0.5, trim= 0in 1in 0.5in 0.5in, clip=true ]{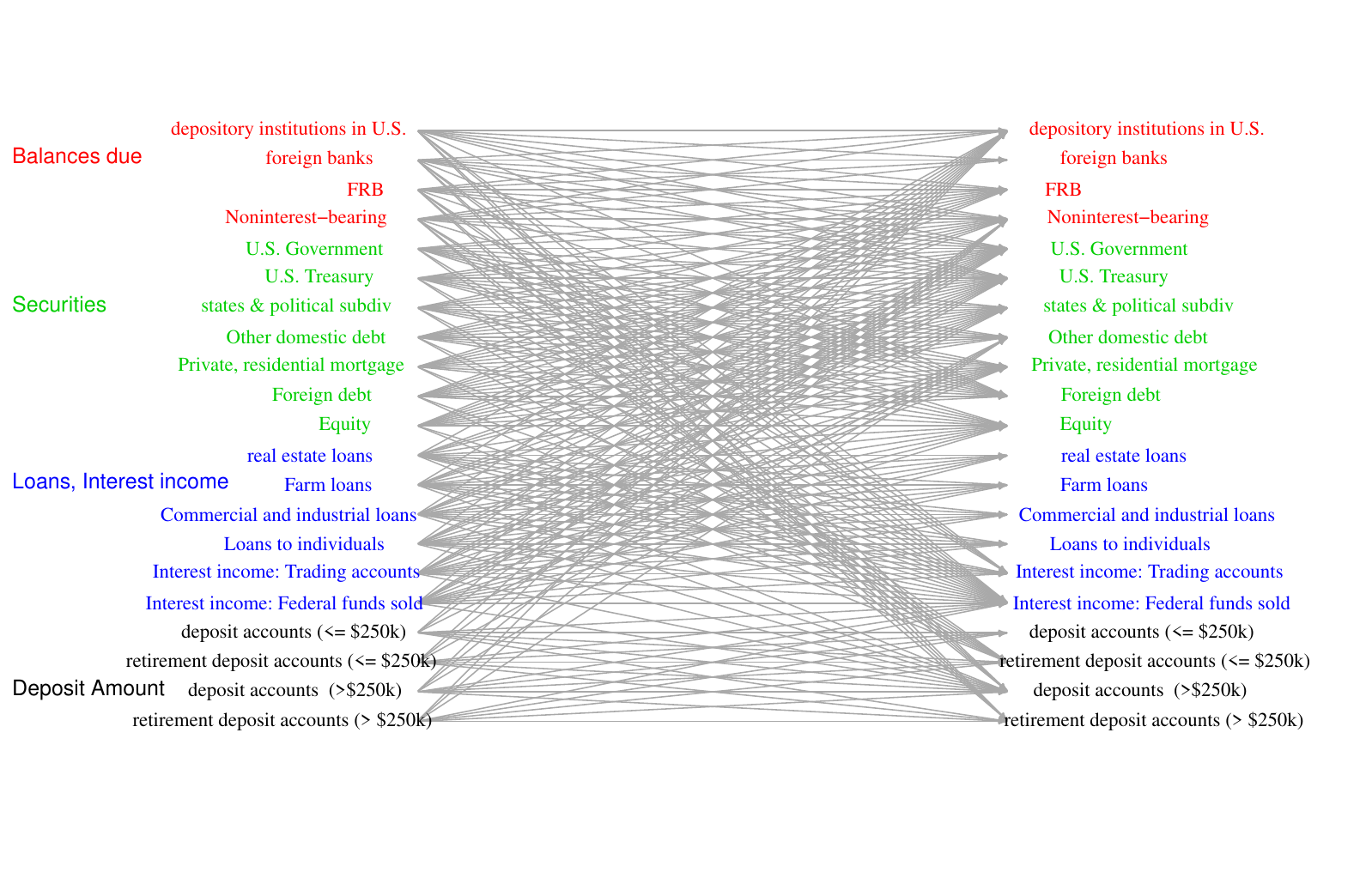}}
\caption{Estimated Networks of banking balance sheet variables using (a) lasso and (b) group lasso. The networks represent the aggregated network over 5 time points.}
\label{fig_bank_l}
\end{figure}
\begin{table}
\caption{Mean and standard deviation (in parentheses) of PMSE (MSE in case of Dec 2010) for prediction of banking balance sheet variables.}
\label{MSE_bank} 
\centering
\begin{tabular}{|r|cccc|} \hline
Quarter & Lasso & Grp & Agrp & Thgrp \\ \hline
Dec 2010 & 1.59 (0.29) & 0.36 (0.05) & 0.36 (0.05) & 0.37 (0.05) \\ 
Mar 2011 & 1.46 (0.30) & 0.47 (0.23) & 0.47 (0.23) & 0.46 (0.22)  \\ 
Jun 2011 & 1.33 (0.26) & 0.36 (0.11) & 0.36 (0.11) & 0.35 (0.11)  \\ 
Sep 2011 & 1.72 (0.32) & 0.50 (0.18)& 0.50 (0.18) & 0.47 (0.16)  \\ \hline
\end{tabular}
\end{table}
In this application, we examine the structure of the balance sheets in terms of assets and liabilities of the $n=50$ largest (in terms of total balance sheet size) US banking corporations. The data cover $9$ quarters (September 2009-September 2011) and were directly obtained from the Federal Deposit Insurance Corporation (FDIC) database (available at \texttt{www.fdic.gov}). The $p=21$ variables correspond to different assets (US and foreign government debt securities, equities, loans (commercial, mortgages), leases, etc.) and liabilities (domestic and foreign deposits from households and businesses, deposits from the Federal Reserve Board, deposits of other financial institutions, non-interest bearing liabilities, etc.) We have organized them into four categories: two for the assets (loans and securities) and two for the liabilities (Balances Due and Deposits, based on a \$250K reporting FDIC threshold). Amongst the 50 banks examined, one discerns large integrated ones with significant retail, commercial and investment activities (e.g. Citibank, JP Morgan, Bank of America, Wells Fargo), banks primarily focused on investment business (e.g. Goldman Sachs, Morgan Stanley, American Express, E-Trade, Charles Schwab), regional banks (e.g. Banco Popular de Puerto Rico, Comerica Bank, Bank of the West).

The raw data are reported in thousands of dollars. The few missing values were imputed using a nearest neighbor imputation method with $k=5$, by clustering them according to their total assets in the most recent quarter (September 2011) and subsequently every missing observation for a particular bank was imputed by the median observation on its five nearest neighbors. The data were log-transformed to reduce non-stationarity issues. The dataset was restructured as a  panel with $p=21$ variables and $n=50$ replicates observed over $T=9$ time points. Every column of replicates was  scaled to have unit variance.

We applied the proposed variants of NGC estimates on the first $T=6$ time points (Sep 2009 - Dec 2010) of the above panel dataset. The parameters $\lambda$ and $\delta_{grp}$ were chosen using a $19:1$ sample-splitting method and the misspecification threshold $\delta_{misspec}$ was set to zero as the grouping structure was reliable.  We calculated the MSE of the fitted model in predicting the outcomes in the four quarters (December 2010 - September 2011). The Predicted MSE (MSE for Dec 2010) are listed in  Table~\ref{MSE_bank}. The estimated network structures are shown in Figure~\ref{fig_bank_l}. 

It can be seen that the lasso estimates recover a very simple temporal structure amongst the
variables; namely, that past values (in this case lag-1) influence present ones. Given the
structure of the balance sheet of large banks, this is an anticipated result, since it can not
be radically altered over a short time period due to business relationships and past commitments
to customers of the bank.
However, the (adaptive) group lasso estimates reveal a richer and
more nuanced structure. Examining the
fitted values of the adjacency matrices $A^t$, we notice that the dominant effects remain those discovered by
the lasso estimates. However, fairly strong effects are also estimated within each group, but
also between the groups of the assets (loans and securities) on the balance sheet. This suggests
rebalancing of the balance sheet for risk management purposes between relatively low risk securities
and potentially more risky loans. Given the period covered by the data (post financial crisis starting in September 2009)
when credit risk management
became of paramount importance, the analysis picks up interesting patterns. On the other hand,
significant fewer associations are discovered between the liabilities side of the balance sheet.
Finally, there exist relationships between deposits and securities such as US Treasuries and other domestic
ones (primarily municipal bonds); the latter indicates that an effort on behalf of the banks to manage the credit risk of
their balance sheets, namely allocating to low risk assets as opposed to more risky loans.

It is also worth noting that the group lasso model exhibits superior predictive performance over the lasso estimates,
even 4 quarters into the future. Finally, in this case the thresholded estimates did not provide any additional benefits
over the regular and adaptive variants, given
that the specification of the groups was based on accounting principles and hence correctly structured.

\section{Discussion}\label{disc}

In this paper, the problem of estimating Network Granger Causal (NGC) models with inherent
grouping structure is studied when replicates are available. Norm, and both group level and within group variable selection consistency are established under fairly mild assumptions on the structure of the underlying time series. To achieve the second objective the novel concept of direction consistency is introduced.

The type of NGC models discussed in this study have wide applicability in different areas, including genomics and economics.
However, in many contexts the availability of replicates at each time point is not feasible (e.g. in rate of returns for stocks or other macroeconomic variables), while grouping structure is still present (e.g. grouping of stocks according to industry sector).
Hence, it is of interest to study the behavior of group lasso estimates in such a setting and address the technical challenges emanating from such a pure time series (dependent) data structure.

%
\appendix
\section{Auxiliary Lemmas}\label{app_lemma}
\begin{lem}[Characterization of the Group lasso estimate]\label{lem_kkt}
A vector  $\hat{\beta} \in \RR^p$ is a solution to the convex optimization problem
\begin{equation}\label{eqn:grplasso}
\argmin_{\beta \in \RR^p} \frac{1}{2n} \left\| Y - X \beta\right\|^2 + \sum_{g=1}^G \lambda_g \| \beta_{[g]} \|
\end{equation}
if and only if $\hat{\beta}$ satisfies, for some $\tau \in \RR^p$ with $\max_{1 \le g \le G} \left\| \tau_\gi \right\| \, \le 1$, 
$\frac{1}{n} \bbb{ X' (Y - X \hat{\beta})}_\gi = \lambda_g \, \tau_\gi \, \forall g$.
Further, $\tau_\gi = D \b{\hat{\beta}_\gi}$ whenever $\hat{\beta}_{[g]} \ne \mathbf{0}$.
\end{lem}
\begin{proof}
Follows directly from the KKT conditions for the optimization problem \eqref{eqn:grplasso}.
\end{proof}
\begin{lem}[Concentration bound for multivariate Gaussian]\label{talag}
Let $Z_{k \times 1} \sim N(0, \Sigma)$. Then, for any $t > 0$, the following inequalities hold:
\begin{equation*}
\PP[| \|Z\| - \EE\|Z \| | > t ] \le 2 \, \exp \left(-\frac{2t^2} {\pi^2 \| \Sigma \|}\right),~~~~\EE \nn{Z} \le \sqrt{k} \sqrt{\nn{\Sigma}}
\end{equation*}
\end{lem}
\begin{proof}
The first inequality can be found in \citet{talagrand} (equation (3.2). To establish the second inequality note that,
\begin{equation*}
\EE {\nn{Z}} \le \sqrt{\EE {\nn{Z}^2}} = \sqrt{\EE \bbb{\tr \b{Z Z'}}} = \sqrt{\tr \b{\Sigma}} \le \sqrt{k} \sqrt{\nn{\Sigma}}
\end{equation*}
\end{proof}

\begin{lem}\label{littlelemma1}
Let $\beta,~\hat{\beta} \in \mathbb{R}^m\backslash \{ \mathbf{0}\}$. Let $\hat{u} = \hat{\beta} - \beta$ and ${r} = D(\hat{\beta}) - D(\beta)$. Then $\left \| r \right \| < 2 \delta$ whenever $\left \| \hat{u}\right \| < \delta  \left \| \beta \right \|$.
\end{lem}
\begin{proof}
It follows from $\left \| \hat{u}\right \| < \delta \left \| \beta \right \|$ that
\begin{equation*}
(1-\delta) \| \beta \| < \| \beta \| - \| \hat{u} \| \le \| \hat{\beta} \| \le \|\hat{u} \| + \| \beta \| < (1+\delta) \| \beta \| \, ,
\end{equation*}
which implies that $\left | \| \beta \| - \| \hat{\beta} \| \right| < \delta \| \beta \|$. Now,
\begin{equation*}
{ \| \hat{\beta} \| \left \| \beta \right \|}\| r \| =  \left \| ~\hat{\beta} \| \beta\| + (\hat{u} - \hat{\beta}) \| \hat{\beta}\| ~ \right \|
\le \left \| \hat{\beta} \left(\| \beta \| - \| \hat{\beta} \| \right)  + \| \hat{\beta} \| ~\hat{u} ~  \right\|
<{ \| \hat{\beta} \| \left \| \beta \right \|} (\delta + \delta)
\end{equation*}
since $\left | \| \beta \| - \| \hat{\beta} \| \right| < \delta \| \beta \|$ and $\left \| \hat{u}\right \| < \delta \left \| \beta \right \|$.
\end{proof}
\section{Proof of Main Results}\label{app_main}
\begin{proof}[Proof of Proposition \ref{spectralresult}]
(a)  Note that $\Sigma$ is a $p(T-1) \times p(T-1)$ block Toeplitz matrix with $(i,j)^{th}$ block  $(\Sigma_{ij})_{1 \leq i, j \leq (T-1)} := \Gamma(i-j)$, where $\Gamma(\ell)_{p \times p}$ is the autocovariance function of lag $\ell$ for the zero-mean VAR(d) process \eqref{eqn1:NGCdefn}, defined as $\Gamma(\ell) = \EE [\mathbf{X}^t (\mathbf{X^{t-\ell}})']$.

We consider the cross spectral density of the VAR(d) process \eqref{eqn1:NGCdefn}
\begin{equation}\label{spectra}
f(\theta) =\frac{1}{2\pi} \displaystyle \sum_{\ell = -\infty}^{\infty} \Gamma(\ell) e^{-i\ell \theta}, ~~ \theta  \in [-\pi, \pi]
\end{equation}
From standard results of spectral theory we know that $\Gamma(\ell) = \int_{-\pi}^\pi e^{i\ell\theta} \, f(\theta) \, d\theta$, for every $\ell$.

We want to find a lower bound on the minimum eigenvalue of $\Sigma$, i.e., $ \inf_{\nn{x} = 1} x' \Sigma x$. Consider an arbitrary $p(T-1)$-variate unit norm vector $x$, formed by stacking the $p$-tuples $x^1, \dots, x^{T-1}$.

 For every $\theta \in [-\pi, \pi]$ define $G(\theta) = \sum_{t=1}^{T-1} x^t \,  e^{-it \theta}$ and note that
\begin{eqnarray*}
\int_{-\pi}^\pi G^* (\theta) G(\theta) \, d\theta &=& \displaystyle \sum_{t=1}^{T-1} \displaystyle \sum_{\tau = 1}^{T-1} (x^t)' (x^\tau) \displaystyle \int_{-\pi}^\pi e^{i(t-\tau) \theta} \, d\theta\\
&=& \displaystyle \sum_{t=1}^{T-1} \displaystyle \sum_{\tau = 1}^{T-1} (x^t)' (x^\tau) \, \, ( 2\pi \,\ind_{\{t=\tau\}})= 2\pi \,  \displaystyle \sum_{t=1}^{T-1}(x^t)' (x^t) = 2 \pi \, \nn{x}^2 = 2\pi
 \end{eqnarray*}
Also let $\mu(\theta)$  be the minimum eigenvalue  of the Hermitian matrix $f(\theta)$. Following  \cite{parter61extemeeigentoeplitz} we have the result
\begin{eqnarray*}
x' \Sigma x &=& \displaystyle \sum_{t=1}^{T-1} \displaystyle \sum_{\tau = 1}^{T-1} (x^t)' \Gamma(t-\tau) x^\tau 
= \displaystyle \sum_{t=1}^{T-1} \displaystyle \sum_{\tau = 1}^{T-1} (x^t)' \b{\int_{-\pi}^\pi e^{i(t-\tau) \theta} f(\theta) d\theta} x^\tau \\
&=& \int_{-\pi}^\pi \b{\displaystyle \sum_{t=1}^{T-1} (x^t)' e^{it\theta}} \, f(\theta) \, \b{\displaystyle \sum_{\tau=1}^{T-1} x^\tau e^{-i \tau \theta}} d\theta 
=  \int_{-\pi}^\pi G^*(\theta) \, f(\theta) \, G(\theta) \, d\theta \\
& \ge&  \int_{-\pi}^\pi \mu(\theta) \b{ G^*(\theta)  G(\theta)} \ d\theta 
\ge \displaystyle \b{\min_{\theta \in (-\pi, \pi)} \mu(\theta)} \,  \int_{-\pi}^\pi G^*(\theta) G(\theta) \, d\theta = 2\pi \, \displaystyle \min_{\theta \in (-\pi, \pi)} \mu(\theta)
\end{eqnarray*}
So $\Lambda_{min} (\Sigma) \ge 2\pi \, \displaystyle \min_{\theta \in (-\pi, \pi)} \mu(\theta)$.
If $A(z) = I - A^1z - A^2 z^2 - \ldots - A^d z^d $ is the (matrix-valued) characteristic polynomial of the VAR(d) model \eqref{eqn1:NGCdefn}, then we have the following representation of the spectral density (see eqn (9.4.23), \cite{Priestley2}):
\begin{equation*}
f(\theta) = \frac{1}{2\pi} \sigma^2 (A(e^{-i \theta}))^{-1} (A^*(e^{-i \theta}))^{-1}
\end{equation*}
Thus, $ 2\pi \mu(\theta) =2\pi  \Lambda_{min} (f(\theta)) = 2\pi/\Lambda_{max} (f(\theta)^{-1}) \ge \sigma^2/ \nn{A(e^{-i\theta})}$. But $\nn{A(e^{-i\theta})} \le 1+ \sum_{t=1}^d \nn{A^t}$ for every $\theta \in [-\pi, \pi]$. The result then follows at once from the standard matrix norm inequality \citep[see e.g.][Cor $2.3.2$]{matrixcomputations},
\begin{equation*}
\| A^t \|_2 \le \sqrt{\|A^t\|_1 \|A^t\|_{\infty}} \le \frac{\|A^t\|_1 + \|A^t\|_{\infty}}{2}~~t=1, \ldots, d
\end{equation*}
since 
\begin{equation*}
\|A^t\|_1  = \displaystyle \max_{1 \le i \le p } \displaystyle \sum_{j=1}^p |A^t_{ij}|,~~
\|A^t\|_{\infty}  = \displaystyle \max_{1 \le j \le p } \displaystyle \sum_{i=1}^p |A^t_{ij}|
\end{equation*}
(b) The first part of the result guarantees that $\Sigma^{1/2}$ satisfies RE(s) for any $s>0$ with $\phi^2_{RE}(\Sigma^{1/2}) \ge \Lambda_{min}(\Sigma) \ge m^2$. The second part of the proposition is a straightforward adaptation of \citep[Cor $1$]{raskutti2010REcorrgauss} tailored for group lasso penalty.
\end{proof}
\begin{proof}[Proof of Theorem \ref{selectconsist}]
Consider any solution $\hat{\beta}_R \in \RR^q$ of the restricted regression
\begin{equation}\label{restrictgrp}
\displaystyle \argmin_{\beta \in \mathbb{R}^q} \frac{1}{2n} \left \| \mathbf{Y} - {X}_{(1)} \beta \right \|_2^2 + \lambda \displaystyle \sum_{g=1}^s \left \| \beta_{[g]} \right \|_2
\end{equation}
and set $\hat{\beta} = \bbb{ \hat{\beta}_R' \, :  \mathbf{0}_{1 \times (p-q)} }' $.  We show that such an augmented vector $\hat{\beta}$ satisfies the statements of Theorem \ref{selectconsist} with high probability.

Let $\hat{u} = \hat{\beta}_{(1)} - \beta^0_{(1)} = \hat{\beta}_R - \beta^0_{(1)}$. In view of lemmas \ref{lem_kkt} and \ref{littlelemma1}, it   suffices to show that the following events happen with probability at least $1 - 4\, G^{1-\alpha}$:
\begin{eqnarray}
&~& \nn{ \hat{u}_\gi } < \delta_n \nn{\beta^0_\gi} \mbox{, for all $g \in S$} \label{(A)}\\
&~& \frac{1}{n} \nn{ \bbb{X' \b{\epsilon - X_{(1)} \hat{u}}}_\gi  } \le \lambda \mbox{, for all $g \notin S$}\label{(B)}
\end{eqnarray}
Note that, in view of Lemma~\ref{lem_kkt},  $\hat{u} = \b{C_{11}}^{-1} \b{\frac{{1}}{\sqrt{n}}Z_{(1)} - \lambda \tau }$ for some $\tau \in \RR^q$ with $\nn{\tau_\gi} \le 1$ for all $g \in S$, and $Z = \frac{1}{\sqrt{n}} X' \epsilon = \bbb{ Z_{(1)}': Z_{(2)}' }'$.
Thus, for any $g \in S$,
\begin{eqnarray*}
\PP \b{ \nn{\hat{u}_\gi} > \delta_n \nn{\beta^0_\gi}  } \le \PP \b{ \nn{ \bbb{\b{C_{11}}^{-1} \b{\frac{{1}}{\sqrt{n}}Z_{(1)} - \lambda \tau }}_\gi } >  \delta_n \nn{\beta^0_\gi} }\\
\le \PP \b{ \nn{ \bbb{\b{C_{11}}^{-1} Z_{(1)}}_\gi } > \sqrt{n} \bbb{\delta_n \nn{{\beta^0}_\gi} - \lambda \nn{\bbb{\b{C_{11}}^{-1}\tau}_\gi}}  }
\end{eqnarray*}
Note that $V = \b{C_{11}} \inv Z_{(1)} \sim N(\mathbf{0}, \sigma^2 \b{C_{11}}\inv)$. So $V_\gi \sim N(\mathbf{0}, \sigma^2  C_{11}^{\gi \gi})$, where $\Sigma^{[g][g]}: = (\Sigma \inv)_{[g][g]}$.
Also, by the second statement of lemma \ref{talag} we have $\EE \nn{V_{[g]}} \le \sigma \sqrt{k_g} \sqrt{\nn{C^{[g][g]}_{11}}}$. 
Therefore $\PP \b{ \nn{\hat{u}_\gi} > \delta_n \nn{\beta^0_\gi} }$ is bounded above by 
\begin{eqnarray*}
\PP \b{ \n{ \nn{V_{[g]}} - \EE \nn{V_\gi} }  > \sqrt{n} \bbb{\delta_n \nn{{\beta^0}_\gi} - \lambda \nn{\b{C_{11}}^{-1}}\sqrt{s}} - \sigma \sqrt{k_g \nn{C_{11}^{[g][g]}}}  }\\
\le 2 \exp \bbb{-\frac{2}{\pi^2 \sigma^2 \|  C^{[g][g]}_{11}\|} \b{  \sqrt{n} \delta_n \| \beta^0_\gi \| - \sqrt{n} \lambda \| C^{-1}_{11} \| \sqrt{s} -\sigma \sqrt{k_g \| C^{[g][g]}_{11}\|} }^2 } \hspace{1.1cm}
\end{eqnarray*}
For the proposed choice of $\delta_n$, this expression is bounded above by $2 \, G^{-\alpha}$.\\
\noindent Next, for any $g \notin S$, we get
\begin{eqnarray*}
\PP \b{\frac{1}{n} \nn{ \bbb{X' \b{\epsilon - X_{(1)} \hat{u}}}_\gi  } > \lambda }  \hspace{7cm}\\   
\le \PP \b{  \nn{\bbb{Z_{(2)} - C_{21} C_{11}\inv Z_{(1)}}_\gi} > \sqrt{n}{\lambda} \b{1 - \nn{\bbb{C_{21} C_{11}\inv \tau}_\gi}}  }
\end{eqnarray*}
Defining $W = Z_{(2)} - C_{21} C_{11}\inv Z_{(1)} \sim N(\mathbf{0},\sigma^2 ( C_{22} - C_{21} C_{11}\inv C_{12}))$,
the uniform irrepresentable condition implies that the above probability is bounded above by $\PP \b{\nn{W_\gi} > \sqrt{n} \lambda \eta}$.

It can then be seen that $W_\gi \sim N(\mathbf{0},\sigma^2  \bar{C}_{[g][g]})$, where $\bar{C} = C_{22} - C_{21} C_{11}\inv C_{12}$ denotes the Schur complement of $C_{22}$.
As before, lemma \ref{talag} establishes that
\begin{eqnarray*}
\PP \b{\nn{W_\gi} > \sqrt{n} \lambda \eta} &\le& \PP \b{\n{ \nn{W_\gi} - {\EE \nn{W_\gi}}}    >    \sqrt{n} \lambda \eta - \sigma \sqrt{k_g \|\bar{C}_{[g][g]} \|}}\\
&\le& 2 \exp \bbb{-\frac{2}{\pi^2 \| \sigma^2 \bar{C}_{[g][g]}\|} \b{\sqrt{n} \lambda \eta - \sigma \sqrt{k_g \|\bar{C}_{[g][g]} \|} }^2 },
\end{eqnarray*}
and the last probability is bounded above by $2G^{-\alpha}$ for the proposed choice of $\lambda$.\\
The results in the proposition follow by considering the union bound on the two sets of the probability statements made across all $g \in \NN_G$.
\end{proof}

%
%

\section{Proof of results on $\ell_2$-consistency}\label{app_l2}
We first note that each of the $p$ optimization problems in (\ref{eqn:NGCest}) is essentially a generic group lasso regression on $n$ independent samples from a linear model $Y = X \beta^0 + \epsilon$, $\epsilon \sim N(0, \sigma^2)$:
\begin{equation}\label{genericgrplasso}
\hat{\beta} =  \displaystyle \argmin_{\beta \in \mathbb{R}^p} \frac{1}{2n} \| \mathbf{Y} - \mathbf{X} \mathbf{\beta} \|^2 +  \displaystyle \sum_{g=1}^{\bar{G}} \lambda_g \|  \mathbf{\beta}_{[g]}\| 
\end{equation}
where $\mathbf{Y}_{n \times 1} = \mathcal{X}^T_i$, $\mathbf{X}_{n \times \bar{p}} = [ \mathcal{X}^1: \cdots: \mathcal{X}^{T-1}]$, $\mathbf{\beta}^0_{\bar{p} \times 1} = vec(A^{1 :(T-1)}_{i:})$,  $\{ 1, \ldots, \bar{p}\} = \displaystyle \cup_{g=1}^{\bar{G}}  \mathcal{G}_g$, $\bar{p} = (T-1)p$, $\bar{G} = (T-1)G$ and $\lambda_g  = \lambda w^t_{i, g}$.
We provide the proofs in the context of a generic group lasso penalized regression problem.

Recall the Restricted Eigenvalue assumption required for the derivation of $\ell_2$ estimation and prediction error. Following \citet{vandegeerconditions}, we introduce a slightly weaker notion called \textbf{Group Compatibility} (GC). For a constant $L>0$ we say that GC(S, L) condition holds, if there exists a constant\\ $\phi_{compatible} = \phi_{compatible}(S, L) > 0$ such that
\begin{equation}\label{compatible}
\min_{\Delta \in \mathbb{R}^p \backslash \{\mathbf{0}\}} \left \{ \frac{\left( \sum_{g \in S} \lambda_g^2 \right)^{1/2} \| X \Delta\| }{\sqrt{n} \displaystyle \sum_{g \in S} \lambda_g \| \Delta_{[g]} \| } : \displaystyle \sum_{g \notin S} \lambda_g \| \Delta_{[g]}\|  \le L \displaystyle \sum_{g \in S} \lambda_g \| \Delta_{[g]} \| \right \} \ge \phi_{compatible}
\end{equation}
The fact that GC(S, L) holds whenever RE(s, L) is satisfied (and $\phi_{RE} \le \phi_{compatible}$) follows at once from Cauchy Schwarz inequality. 
We shall derive upper bounds on the prediction and $\ell_{2,1}$ estimation error of group lasso estimates involving the compatibility constant. This notion will also be used later to connect the irrepresentable conditions to the consistency results of group lasso estimators.

\begin{prop}\label{compatible2consist}
Suppose the GC condition \eqref{compatible} holds with $L=3$. Choose $\alpha > 0$ and denote $\lambda_{min} = \min_{1 \le g \le G} \lambda_g$. If
$$
\lambda_g \ge \frac{2 \sigma}{\sqrt{n}} \sqrt{\nn{C_{[g] [g]}}} \left(  \sqrt{k_g} + \frac{\pi}{\sqrt{2}} \sqrt{\alpha \, \log\,G} \right)
$$
for every $g \in \mathbb{N}_{G}$, then, the following statements hold with probability at least $1 - 2G^{1-\alpha}$,
\begin{eqnarray}
&~& \frac{1}{n} \left\| X \left(\hat{\beta} - \beta^0 \right) \right\|^2 \le \frac{16}{\phi^2_{compatible}} \sum_{g=1}^s \lambda^2_g\label{eqpred}\\
&~& \| \hat{\beta} - \beta^0 \|_{2, 1} \le \frac{16}{\phi_{compatible}^2} \,  \frac{\sum_{g=1}^s \lambda_g ^2}{\lambda_{min}}. \label{eql21}
\end{eqnarray}
If, in addition, RE(2s, 3) holds, then, with the same probability we get
\begin{equation}\label{re2stol2consist}
\| \hat{\beta} - \beta^0 \| \le \frac{4\sqrt{10}}{\phi_{RE}^2 (2s)}  \, \frac{\sum_{g=1}^s \lambda_g^2}{\lambda_{min} \, \sqrt{s }} \, .
\end{equation}
\end{prop}

\begin{proof}[Proof of Proposition \eqref{compatible2consist}] Since $\hat{\beta}$ is a solution of the optimization problem \eqref{genericgrplasso},  for all $\beta \in \RR^p$, we have
\begin{equation*}
\frac{1}{n} \| Y - X\hat{\beta} \|^2 + 2 \sum_{g=1}^{G} \lambda_g \|\hat{\beta}_{[g]} \| \le \frac{1}{n} \| Y - X{\beta} \|^2 + 2 \sum_{g=1}^{G} \lambda_g \|{\beta}_{[g]} \|.
\end{equation*}
Plugging in $Y = X \beta^0 + \epsilon$, and simplifying the resulting equation, we get
\begin{eqnarray*}
&~& \frac{1}{n} \| X (\hat{\beta} - \beta^0) \|^2 \le \frac{1}{n} \| X(\beta - \beta^0) \|^2 + \frac{2}{n} \sum_{g=1}^G \left\| (X' \epsilon)_{[g]}\right\| \left\| (\hat{\beta} - \beta )_{[g]}\right\| \\
&~&\hspace{1.5in}+ 2 \sum_{g=1}^G \lambda_g \left(\| \beta_{[g]}  \|- \| \hat{\beta}_{[g]}\| \right).
\end{eqnarray*}
Fix $g \in \mathbb{N}_G$ and consider the event $\calA_g = \bb{\epsilon \in \RR^n : \frac{2}{n} \nn{\b{X' \epsilon}_\gi} \le \lambda_g}$. Note that $Z = \frac{1}{\sqrt{n}} X' \epsilon \sim N(\mathbf{0},\sigma^2  C)$. So $Z_\gi \sim N(\mathbf{0}, \sigma^2 C_{\gi \gi})$. Then,
\begin{eqnarray*}
\PP \b{\calA^c_g} &=& \PP \b{\nn{Z_\gi} > \frac{1}{2} \lambda_g \sqrt{n}}\\
&\le& \PP \b{\n{Z_\gi - \E \nn{Z_\gi}} > \frac{\lambda_g \sqrt{n}}{2} - \sigma \sqrt{k_g} \sqrt{\nn{C_{\gi \gi}}}},
\end{eqnarray*}
where the last inequality follows from the second statement of Lemma~\ref{talag}.
Now, let $x_g = \frac{\lambda_g \sqrt{n}}{2} - \sigma \sqrt{k_g} \sqrt{\nn{C_{\gi \gi}}}$. Then, for $x_g > 0$, if
$$
2 \exp \b{-\frac{2 \, x_g^2}{\pi^2 \sigma^2 \nn{C_{\gi \gi}}}} \le 2\, G^{-\alpha} \, ,
$$
we get
$$
\PP \b{\calA^c_g} \le 2G^{-\alpha}.
$$
But this happens if,
$$
\sqrt{2} x_g \ge \sqrt{\alpha \, \log\,G} \pi \sigma \sqrt{ \nn{C_{\gi \gi}}},
$$
which is ensured by the proposed choice of $\lambda_g$.

Next, define $\calA := \cap_{g=1}^G \calA_g$. Then, $\PP \b{\calA} \ge 1  -2 G^{1-\alpha}$, and on the event $\calA$, we have, for all $\beta \in \RR^p$,
\begin{eqnarray}
 &~& \frac{1}{n} \| X (\hat{\beta} - \beta^0) \|^2 + \sum_{g=1}^G \lambda_g \nn{\hat{\beta}_{\gi } - \beta_{\gi }} \le \frac{1}{n} \| X(\beta - \beta^0) \|^2  \nonumber \\
&~& \hspace{1.5in} + 2 \sum_{g=1}^G \lambda_g \b{ \nn{\hat{\beta}_{\gi} - \beta_{\gi}} + \nn{\beta_{\gi}} - \nn{\hat{\beta}_{\gi}} }. \nonumber
\end{eqnarray}
Note that $\b{ \nn{\hat{\beta}_{\gi} - \beta_{\gi}} + \nn{\beta_{\gi}} - \nn{\hat{\beta}_{\gi}} }$ vanishes if $g \notin S$ and is bounded above by \\ $\min \{ 2 \nn{\beta_{\gi}}, 2 \b{\nn{\beta_\gi - \hat{\beta}_\gi}}\}$ if $g \in S$.

This leads to the following sparsity oracle inequality, for all $\beta \in \RR^p$,
\begin{eqnarray}\label{oracle}
&~& \frac{1}{n} \| X (\hat{\beta} - \beta^0) \|^2 + \sum_{g=1}^G \lambda_g \nn{\hat{\beta}_{\gi} - \beta_{\gi}} \le \frac{1}{n} \| X(\beta - \beta^0) \|^2  \nonumber \\
&~& \hspace{1.5in} + 4 \sum_{g \in S} \lambda_g \min \bb{{ \nn{ \beta_{\gi}} , \, \nn{\beta_{\gi} - \hat{\beta}_{\gi}} } }.
\end{eqnarray}


The sparsity oracle inequality \eqref{oracle} with $\beta = \beta^0$, and $\Delta := \hat{\beta} -\beta^0$ leads to the following two useful bounds on the prediction and $\ell_{2,1}$-estimation errors:
\begin{eqnarray}\label{keyineq}
\frac{1}{n} \nn{ X \Delta}^2 \le 4\sum_{g \in S} \lambda_g \nn{\Delta_{\gi}} \\
\sum_{g \notin S} \lambda_g \nn{\Delta_\gi} \le 3 \sum_{g \in S} \lambda_g \nn{\Delta_\gi}.
\end{eqnarray}
Now, assume the group compatibility condition~\ref{compatible} holds. Then,
\begin{eqnarray}\label{tuktak}
\frac{1}{n} \nn{X \Delta}^2 \le 4 \sum_{g \in S} \lambda_g \nn{\Delta_\gi} \le \sqrt{\sum_{g \in S} \lambda_g^2} \, \frac{\nn{X \Delta}}{\sqrt{n}}\, \frac{4}{\phi_{compatible}},
\end{eqnarray}
which implies the first inequality of proposition \ref{compatible2consist}. The second inequality follows from
\begin{eqnarray*}
\lambda_{min} \nn{\hat{\beta} -\beta}_{2, 1} &\le& \sum_{g=1}^G \lambda_g \nn{\Delta_\gi} \le 4 \sum_{g \in S} \lambda_g \nn{\Delta_\gi} \\
& \le&  4 \sqrt{\sum_{g \in S} \lambda_g^2} \, \frac{\nn{X \Delta}}{\sqrt{n}} \, \frac{1}{\phi_{compatible}} \le \frac{16} {\phi^2_{compatible}} \sum_{g \in S} \lambda_g^2\, ,
\end{eqnarray*}
where the last step uses \eqref{tuktak}.

The proof of the last inequality of proposition \ref{compatible2consist}, i.e., the upper bound on $\ell_2$ estimation error under $RE(2s)$, is the same as in Theorem 3.1 in \citet{lounici2011} and is omitted.
\end{proof}

\section{Irrepresentable assumptions and consistency}\label{app_selection}
In this section, we discuss two results involving the compatibility and irrepresentable conditions for group lasso. We first show that a stronger version of the uniform irrepresentable assumption implies the group compatibility \eqref{compatible}, and hence, consistency in $\ell_{2,1}$ norm. Next we argue that a weaker version of the irrepresentable assumption is indeed necessary for the direction consistency of the group lasso estimates. These results generalize analogous properties of lasso \citep{vandegeerconditions,Zhaoyu06} to the group penalization framework. The proofs are given under a special choice of tuning parameter $\lambda_g = \lambda \sqrt{k_g}$. Similar results can be derived for the general choice of $\lambda_g$, although their presentation is more involved.

\begin{prop}\label{irep2compatible}
Assume uniform irrepresentable condition \eqref{unifirrep} holds with $\eta  \in (0, 1)$, and $\Lambda_{min}(C_{11}) > 0$. Then  group compatibility(S, L) \eqref{compatible} condition holds whenever $L < \frac{1}{1-\eta}$.
\end{prop}
\begin{proof}
First note that with the above choice of $\lambda_g$ the Group Compatibility $(S, L)$ condition simplifies to
\begin{equation}\label{compatible2}
 \phi_{compatible} := \min_{\Delta \in \mathbb{R}^p \backslash \{\mathbf{0}\}} \left \{ \frac{\sqrt{q} \| X \Delta\| }{\sqrt{n} \displaystyle \sum_{g \in S} \sqrt{k_g} \| \Delta_{[g]} \| } : \displaystyle \sum_{g \notin S} \sqrt{k_g} \| \Delta_{[g]}\|  \le L \displaystyle \sum_{g \in S} \sqrt{k_g} \| \Delta_{[g]} \| \right \} > 0
\end{equation}

Also, the uniform irrepresentable condition guarantees that there exists
$0< \eta<1$ such that $\forall \tau \in \mathbb{R}^q$ with $\| \tau \|_{2, \infty} = \displaystyle \max_{1 \le g \le s} \| \tau_{[g]} \|_2 \le 1$, we have,
\[
\frac{1}{\sqrt{k_g}} \left \|  \left[ C_{21} \left( C_{11}\right)^{-1} K^{0} \tau \right]_{[g]} \right \|_2 < 1-\eta ~\forall g \notin S
\]
Here $K^{0} = K/\lambda$ is a $q \times q$ block diagonal matrix with $s$ diagonal blocks $\sqrt{k_1} \, \mathbf{I}_{k_1 \times k_1}, \ldots, \sqrt{k_s} \, \mathbf{I}_{k_s \times k_s}$.
Define
\begin{equation}\label{minbeta}
\Delta^{0}:= \displaystyle \argmin_{\Delta \in \mathbb{R}^p} \left \{  \frac{1}{2n} \| \mathbf{X}\Delta \|^2_2 :~~\displaystyle \sum_{g \in S} \sqrt{k_g} \| \Delta_{[g]} \|_2 = 1,~\displaystyle \sum_{g \notin S} \sqrt{k_g} \| \Delta_{[g]} \|_2 \le L \right \}
\end{equation}
Note that $\frac{1}{n} \| \mathbf{X} \Delta^{0}\|_2^2 = \phi^{2}_{compatible}/q$, and introduce two Lagrange multipliers $\lambda$ and $\lambda'$ corresponding to the equality and inequality constraints for solving the optimization problem in \eqref{minbeta}. Also, partition $\Delta^{0} = \left[ {\Delta}^0_{(1)}: {\Delta}^0_{(2)}\right]$ and $\mathbf{X}= \left[ \mathbf{X}_{(1)}: \mathbf{X}_{(2)}\right]$ into signal and nonsignal parts as in \eqref{signonsigpartn}. The first $q$ linear equations of the KKT conditions imply that there exists $\tau^0 \in \mathbb{R}^q $  such that
\begin{eqnarray}
C_{11} \Delta^0_{(1)} + C_{12} \Delta^{0}_{(2)} &=& \lambda K^{0} \tau^{0} \label{kkt_compatible}
\end{eqnarray}
and, for every $g \in S$,
\begin{eqnarray*}
\tau^{0}_{\gi} &=& D(\Delta^0_{[g]}) \mbox{ if } \Delta^0_\gi \ne \mathbf{0}\\
\| \tau^0_\gi \| _2 &\le& 1 \mbox{ if } \Delta^0_\gi = \mathbf{0}
\end{eqnarray*}
It readily follows that ${(\tau^0)}^T K^{0} \Delta^0_{(1)} = \displaystyle \sum_{g \in S} \sqrt{k_g} \| \Delta^0_\gi \|_2 = 1$.\\
Multiplying both sides of \eqref{kkt_compatible} by $(\Delta^0_{(1)})^T$ we get
\begin{equation}\label{eq1}
{\left(\Delta^0_{(1)}\right)}^T C_{11} \Delta^0_{(1)} + {\left(\Delta^0_{(1)}\right)}^TC_{12} \Delta^{0}_{(2)} = \lambda
\end{equation}
Also, \eqref{kkt_compatible} implies
\begin{equation}\label{eq2}
 \Delta^0_{(1)} + \left(C_{11} \right)^{-1} C_{12} \Delta^{0}_{(2)} = \lambda\left(C_{11} \right)^{-1}  K^{0} \tau^{0}
\end{equation}
Multiplying both sides of the equation by $\left(K^{0} \tau^{0}\right)^{T} = \left( \tau^{0}\right)^{T} K^{0}$ we obtain
\begin{equation}\label{eq3}
1 = -\left(\tau^{0} \right)^{T} K^{0} \left(C_{11} \right)^{-1} C_{12} \Delta^{0}_{(2)} + \lambda \left(K^{0} \tau^{0}\right)^{T} \left(C_{11} \right)^{-1} \left(K^{0} \tau^{0}\right)
\end{equation}
Note that the absolute value of the first term,
\begin{eqnarray}\label{eq4}
\left|  \displaystyle \sum_{g \notin S} {\left(\Delta^{0}_\gi \right)}^{T} \left[ C_{21} {\left(C_{11}\right)}^{-1} K^{0} \tau^0 \right]_\gi   \right|,
\end{eqnarray}
is bounded above by
\begin{eqnarray}\label{eq4.5}
(1-\eta) \left( \displaystyle \sum_{g \notin S} \sqrt{k_g} \|\Delta^0_{[g]}\|_2 \right) \le (1-\eta)L
\end{eqnarray}
by virtue of the uniform irrepresentable condition and the Cauchy-Schwartz inequality.\\
Assuming the minimum eigenvalue of $C_{11}$, i.e.,  $\Lambda_{min} \left( C_{11}\right) $, is positive and considering $\| K^{0} \tau^{0} \|_2 \le \sqrt{q}$, the second term is at most $ \lambda \,{q}/\Lambda_{min} \left( C_{11}\right) $. So \eqref{eq3} implies
\begin{equation}\label{eq5}
1 \le (1 -\eta) L + \frac{\lambda q}{\Lambda_{min} \left( C_{11}\right)}
\end{equation}
In particular, $\lambda \ge \Lambda_{min} \left( C_{11}\right)  \left(1 - (1-\eta)L \right)/q $ is positive whenever $L < 1/ (1-\eta)$.\\
Next, multiply both sides of \eqref{eq2} by $( \Delta^0_{(2)})^{T} C_{21}$ to get
\begin{equation}\label{eq6}
\left( \Delta^0_{(2)}\right)^{T} C_{21} \Delta^0_{(1)} + \left( \Delta^0_{(2)}\right)^{T} C_{21}\left(C_{11} \right)^{-1} C_{(12)} \Delta^{0}_{(2)} = \lambda \left( \Delta^0_{(2)}\right)^{T} C_{21}\left(C_{11} \right)^{-1}  K^{0} \tau^{0}
\end{equation}
Using the upper bound in \eqref{eq4.5}, the right hand side is at least $-\lambda (1-\eta)L$.\\
Also a simple consequence of the block inversion formula of the non-negative definite matrix $C$ guarantees that the matrix $C_{22} - C_{21} \left( C_{11}\right)^{-1} C_{12} \mbox {}$  is non-negative definite. Hence,
\begin{eqnarray*}
&~&\left( \Delta^0_{(2)}\right)^{T}\left[C_{22} -  C_{21}\left(C_{11} \right)^{-1} C_{12}\right] \Delta^{0}_{(2)} \ge 0\\
& \mbox{ and } & \left( \Delta^0_{(2)}\right)^{T} C_{22}  \Delta^{0}_{(2)} \ge \left( \Delta^0_{(2)}\right)^{T}  C_{21}\left(C_{11} \right)^{-1} C_{12}\Delta^{0}_{(2)}
\end{eqnarray*}
Putting all the pieces together we get
\begin{eqnarray*}
\phi^2_{compatible}/q &=& \frac{1}{n} \| \mathbf{X} \Delta^0\|_2^2 \\
&=& {\Delta^0_{(1)}}^{T} C_{11} \Delta^0_{(1)} + 2 {\Delta^0_{(2)}}^{T} C_{21} \Delta^0_{(1)} + {\Delta^0_{(2)}}^{T} C_{22} \Delta^0_{(2)}\\
&=& \lambda + {\Delta^0_{(2)}}^{T} C_{21} \Delta^0_{(1)} + {\Delta^0_{(2)}}^{T} C_{22} \Delta^0_{(2)}  \mbox{ {, by \eqref{eq1}}} \\
&\ge&  \lambda - \lambda(1-\eta)L  \mbox{ , by \eqref{eq6}} \\
&=& \lambda(1-(1-\eta)L)
\end{eqnarray*}
Plugging in the lower bound for $\lambda$ we obtain the result; namely,
\[
\phi^2_{compatible} = \Lambda_{min} (C_{11}) \left(1 - (1-\eta)L \right)^2 > 0
\]
for any $L < \frac{1}{1-\eta}$.
\end{proof}

In this section we investigate the necessity of irrepresentable assumptions for direction consistency of group lasso estimates. To this end we first introduce the notion of weak irrepresentability.

For a $q$-dimensional vector $\tau$ define the stacked direction vector
$\underbrace{\tilde{D}(\mathbf{\tau})}_{q \times 1} = [\underbrace{D({\mathbf{\tau}}_{[1]})'}_{k_1 \times 1}, \ldots,  \underbrace{D({\mathbf{\tau}}_{[s]})'}_{k_s \times 1} ]'$.\\
\textbf{Weak Irrepresentable Condition} is satisfied if  
\begin{equation} \label{weakirrep}
\frac{1}{\lambda_g}\left \| \left[ C_{21} {\left(C_{11}\right)}^{-1} K \tilde{D}(\beta^0_{(1)})  \right]_{[g]} \right \| \le 1, ~ \forall g \notin S = \{ 1, \ldots, s\}
\end{equation}

We argue the necessity of weak irrepresentable condition for group sparsity selection and direction consistency under two regularity conditions on the design matrix, as $n,\, p \rightarrow \infty$:\\
\textbf{(A1)} The minimum eigenvalue of the signal part of the Gram matrix, viz. $\Lambda_{min} (C_{11})$, is bounded away from zero.\\
\textbf{(A2)} The matrices $C_{21}$ and $C_{22}$ are bounded above in spectral norm. 

As in the last proposition, we set $\lambda_g = \lambda \, \sqrt{k_g}$ and $K^0 = K/\lambda$.
Suppose that the weak irrepresentable condition does not hold, i.e., for some $g \notin S$  and $\xi > 0$, we have,
\[
\frac{1}{\sqrt{k_g}}\left \| \left[ C_{21} {\left(C_{11}\right)}^{-1} K^{0} \tilde{D}(\beta^0_{(1)})  \right]_{[g]} \right \| > 1+\xi
\]
for infinitely many $n$. Also suppose that there exists a sequence of positive reals $\delta_n \rightarrow 0$ such that the event
\[
E_n := \{ \| D(\mathbf{\hat{\beta}}_{[g]}) - D(\mathbf{\beta}_{[g]}) \| _2 < \delta_n,~\forall g \in S, \mbox{ and } \hat{\beta}_{[g]} = \mathbf{0} \, \forall \, g \notin S \}
\]
satisfies $\mathbb{P} (E_n) \rightarrow 1$ as $p, \, n \rightarrow \infty$.

Note that for large enough $n$ so that $\delta_n < \min_g \| D({\beta_{[g]}}) \|$, we have $\hat{\beta}_{[g]} \neq \mathbf{0}, \, \forall \, g \in S$ on the event $E_n$.

Then, as in the proof of Theorem \ref{selectconsist}, we have, on the event $E_n$,
\begin{eqnarray}
&~& \hat{\mathbf{u}} = {\left( C_{11}\right)}^{-1} \left[ \frac{1}{\sqrt{n}} {Z}_{(1)} - {\lambda} K^0 \tilde{D}(\hat{\beta}_{(1)}) \right]\label{eqn1} \\
& \mbox{and} & \frac{1}{n} \left \| \left[ {\mathbf{X}_{(2)}}^T (\epsilon - \mathbf{X}_{(1)} \hat{{u}}) \right]_{[g]} \right \| \le \lambda \sqrt{k_g},~\forall g \notin S \label{eqn2}
\end{eqnarray}
Substituting the value of $\hat{u}$ from \eqref{eqn1} in \eqref{eqn2}, we have, on the event $E_n$,
\begin{equation*}
\frac{1}{\sqrt{n}} \left \|  \left[  \mathbf{Z}_{(2)} - C_{21} {\left(C_{11} \right)}^{-1} \mathbf{Z}_{(1)} + {\lambda}{\sqrt{n}}C_{21} {\left( C_{11}\right)}^{-1} K^0 \tilde{D}(\hat{\beta}_{(1)}) \right]_{[g]} \right \| \le  {\lambda}\sqrt{k_g},
\end{equation*}
which implies that
\begin{eqnarray} \nonumber
&~& \left \|  \left[ {Z}_{(2)} - C_{21} \left( C_{11}\right)^{-1} {Z}_{(1)} \right]_{[g]} \right \| \label{eqn3.1} \hspace{2cm} \\
&~& \hspace{2cm} \ge {\lambda} \, \sqrt{n}\,  \sqrt{k_g} \left[ \frac{1}{\sqrt{k_g}} \left\| \left[ C_{21} {\left( C_{11}\right)}^{-1} K^0 \tilde{D}(\hat{\beta}_{(1)}) \right]_{[g]} \right\|  - 1 \right].
\end{eqnarray}
Now note that for large enough $n$, if $\| C_{21} \|$ is bounded above, direction consistency
guarantees that the expression on the right is larger than
\[
\frac{1}{2} \, {\lambda} \,  \sqrt{n}\,  \sqrt{k_g} \left[  \frac{1}{\sqrt{k_g}} \left\| \left[ C_{21} {\left( C_{11}\right)}^{-1} K^0 \tilde{D}({\beta}_{(1)}) \right]_{[g]} \right\|  - 1 \right]
\]
which in turn is larger than $\frac{1}{2} \, {\lambda} \,  \sqrt{n}\,  \sqrt{k_g} \, \xi$, in view of the weak irrepresentable condition.

This contradicts  $\PP (E_n) \rightarrow 1 $, since the left-hand side of \eqref{eqn3.1} corresponds to the norm of a zero mean Gaussian random variable with bounded variance structure $\left[ C_{22} - C_{21} (C_{11})^{-1} C_{12} \right]_{[g] [g]}$  while ${\lambda} \,  \sqrt{n}\,  \sqrt{k_g}$ diverges with $\sqrt{\log \, G}$.


\section{Thresholding Group Lasso Estimates.}\label{app_thres}
\begin{proof}[Proof of Theorem \ref{propthres}]
We use the notations developed in the proof of Proposition~\ref{compatible2consist}. First note that, $(ii)$ follows directly from Theorem \ref{selectconsist}. For $(i)$, since the falsely selected groups are present after the initial thresholding, we get $\| \hat{\beta_{[g]}} \| > 4 \lambda$ for every such group. Next, we obtain an upper bound for the number of such groups. Specifically, denoting $\Delta = \hat{\beta} - \beta^0$, we get
\begin{equation}\label{threseq1}
\left| \hat{S} \backslash S \right| \le \frac{\|\hat{\beta}_{S^c}\|_{2,1}}{4 \lambda} = \frac{\sum_{g \notin S} \| \Delta_{[g]} \|}{4 \lambda}.
\end{equation}

Next, note that from the sparsity oracle inequality \eqref{keyineq}, the following holds on the event $\calA$,
\[
\sum_{g \notin S} \|\Delta_{[g]} \| \le 3 \sum_{g \in S} \|\Delta_{[g]} \|
\]
It readily follows that
\[
4 \sum_{g \notin S} \| \Delta_{[g]} \| \le 3 \| \Delta \|_{2,1} \le \frac{48}{\phi^2} \, s \lambda
\]
where the last inequality follows from the $\ell_{2,1}$-error bound of \eqref{eql21}. Using this inequality together with \eqref{threseq1} gives the result.
\end{proof}

\vskip 0.2in
\bibliographystyle{asa}
\bibliography{biblio}

\end{document}